\documentclass{llncs}
\usepackage{amsmath, amssymb}
\usepackage{inlinenum}
\usepackage{bussproofs}
\usepackage{xspace}
\usepackage{cmll}
\usepackage{stmaryrd}
\usepackage{fmtcount}
\usepackage{ifthen}
\usepackage{tikz}
\usepackage{graphicx}
\usepackage{footnote}
\usepackage{wrapfig}
\usepackage{caption}

\newcommand{\dl}{\textsf{d{\kern-0.1em}$\mathcal{L}$}\xspace}
\newcommand{\qdl}{Q\textsf{d{\kern-0.1em}$\mathcal{L}$}\xspace}
\newcommand{\qdtl}{QdTL\xspace}

\newcommand{\inferencerule}[2]{\displaystyle\dfrac {#1}{#2}}

\title{Quantified Differential Temporal Dynamic Logic for Verifying Properties of Distributed Hybrid Systems\thanks{ This material is based upon work supported by the National Science Foundation under NSF
CAREER Award CNS-1054246, Grant Nos. CNS-0926181, CNS-0931985, and CNS-1035800.}}
\author{Ping Hou}
 \institute {Computer Science Department, Carnegie Mellon University}

\begin{document}
\maketitle

\begin{abstract}
   We combine quantified differential dynamic logic (\qdl) for reasoning about the possible behavior of distributed hybrid systems with temporal logic for reasoning about the temporal behavior during their operation. Our logic supports verification of temporal and non-temporal properties of distributed hybrid systems and provides a uniform treatment of discrete transitions, continuous evolution, and dynamic dimensionality-changes. For our combined logic, we generalize the semantics of dynamic modalities to refer to hybrid traces instead of final states. Further, we prove that this gives a conservative extension of \qdl for distributed hybrid systems. On this basis, we provide a modular verification calculus that reduces correctness of temporal behavior of distributed hybrid systems to non-temporal reasoning, and prove that we obtain a complete axiomatization relative to the non-temporal base logic \qdl. Using this calculus, we analyze temporal safety properties in a distributed air traffic control system where aircraft can appear dynamically.

\end{abstract}

\section{Introduction}\label{sec:1}
Ensuring correct functioning of cyber-physical systems is among the most challenging and most important problems in computer science, mathematics, and engineering. Hybrid systems are common mathematical models for cyber-physical systems with interacting discrete and continuous behavior~\cite{DavorenN00,Henzinger96}. Their behavior combines continuous evolution (called {\em flow}) characterized by differential equations and discrete jumps. However, not all relevant cyber-physical systems can be modeled as hybrid systems. Hybrid systems cannot represent physical control systems that are distributed or form a multi-agent system, e.g., distributed car control systems~\cite{HsuESV91} and distributed air traffic control systems~\cite{DowekMC05}. Such systems form {\em distributed hybrid systems}~\cite{DeshpandeGV96,KratzSPL06,DBLP:conf/csl/Platzer10,DBLP:conf/csl/Platzer10:TR} with discrete, continuous, structural, and dimension-changing dynamics. Distributed hybrid systems combine the challenges of hybrid systems and distributed systems. Correctness of safety-critical real-time and distributed hybrid systems depends on a safe operation throughout {\em all} states of all possible trajectories, and the behavior at intermediate states is highly relevant~\cite{AlurCD90,DammHO03,DavorenN00,FaberM06,Henzinger96}.

{\em Temporal logics} (TL) use temporal operators to talk about intermediate states~\cite{AlurCD90,EmersonC82,EmersonH86,Pnueli77}. They have been used successfully in model checking~\cite{AlurCD90,ClarkeGP99,Henzinger96,HenzingerNSY92,MysorePM05} of finite-state system abstractions. State spaces of distributed hybrid systems, however, often do not admit equivalent finite-state abstractions~\cite{Henzinger96,MysorePM05}. Instead of model checking, TL can also be used deductively to prove validity of formulas in calculi~\cite{DavorenCM04,DavorenN00}. Valid TL formulas, however, only express very generic facts that are true for all systems, regardless of their actual behavior. Hence, the behavior of a specific system first needs to be axiomatized declaratively to obtain meaningful results. Then, however, the correspondence between actual system operations and a declarative temporal representation may be questioned.

Very recently, a dynamic logic, called {\em quantified differential dynamic logic} (\qdl) has been introduced as a successful tool for deductively verifying distributed hybrid systems~\cite{DBLP:conf/csl/Platzer10,DBLP:conf/csl/Platzer10:TR}. \qdl can analyze the behavior of actual distributed hybrid system models, which are specified operationally. Yet, operational distributed hybrid system models are {\em internalized} within \qdl formulas, and \qdl is closed under logical operators.
However, \qdl only considers the behavior of distributed hybrid systems at final states, which is insufficient for verifying safety properties that have to hold all the time.

We close this gap of expressivity by combining \qdl with temporal logic~\cite{EmersonC82,EmersonH86,Pnueli77}.
In this paper, we introduce a logic, called {\em quantified differential temporal dynamic logic} (\qdtl), which provides modalities for quantifying over traces of distributed hybrid systems based on \qdl. We equip \qdtl with temporal operators to state what is true all along a trace or at some point during a trace. In this paper, we modify
the semantics of the dynamic modality $[\alpha]$ to refer to all {\em traces} of $\alpha$ instead of all final states
reachable with $\alpha$ (similarly for $\langle \alpha \rangle$). For instance, the formula $[\alpha]\Box\phi$ expresses
that $\phi$ is true at each state during all traces of the distributed hybrid system $\alpha$. With this,
\qdtl can also be used to verify temporal statements about the behavior of $\alpha$ at intermediate states during system runs.
As in our non-temporal dynamic logic \qdl~\cite{DBLP:conf/csl/Platzer10,DBLP:conf/csl/Platzer10:TR}, we use {\em quantified hybrid programs} as an operational model for distributed hybrid systems, since they admit a uniform compositional treatment of interacting discrete transitions, continuous evolutions, and structural/dimension changes in logic.

As a semantical foundation for combined temporal dynamic formulas, we introduce a hybrid trace semantics for \qdtl. We prove that \qdtl is a conservative extension of \qdl: for non-temporal specifications, trace semantics is equivalent to the non-temporal transition semantics of \qdl~\cite{DBLP:conf/csl/Platzer10,DBLP:conf/csl/Platzer10:TR}.

As a means for verification, we introduce a sequent calculus for \qdtl that successively reduces temporal statements about traces of quantified hybrid programs to non-temporal \qdl formulas. In this way, we make the intuition formally precise that temporal safety properties can be checked by augmenting proofs with appropriate assertions about intermediate states. Like in~\cite{DBLP:conf/csl/Platzer10,DBLP:conf/csl/Platzer10:TR}, our calculus works compositionally. It decomposes correctness statements about quantified hybrid programs structurally into corresponding statements about its parts by symbolic transformation.

Our approach combines the advantages of \qdl in reasoning about the behaviour of operational distributed hybrid system models with those of TL to verify temporal statements about traces. We show that \qdtl is sound and relatively complete. We argue that \qdtl can verify practical systems and demonstrate this by studying temporal safety properties in distributed air traffic control. Our primary contributions are as follows:
\begin{itemize}
\item We introduce a logic for specifying and verifying temporal properties of distributed hybrid systems.
\item We present a proof calculus for this logic, which, to the best of our knowledge, is the first verification approach that can handle temporal statements about distributed hybrid systems.
\item We prove that this compositional calculus is a sound and complete axiomatization relative to differential equations.
\item We verify temporal safety properties in a collision avoidance maneuver in distributed air traffic control, where aircraft can appear dynamically.
\end{itemize}

\section{Related Work}\label{sec:2}
Multi-party distributed control has been suggested for car control~\cite{HsuESV91} and air traffic control~\cite{DowekMC05}. Due to limits in verification technology, no formal analysis of temporal statements about the distributed hybrid dynamics has been possible for these systems yet. Analysis results include discrete message handling~\cite{HsuESV91} or collision avoidance for two participants~\cite{DowekMC05}.

The importance of understanding dynamic/reconfigurable distributed hybrid systems was recognized in modeling languages SHIFT~\cite{DeshpandeGV96} and R-Charon~\cite{KratzSPL06}. They focused on simulation/compilation~\cite{DeshpandeGV96} or the development of a semantics~\cite{KratzSPL06}, so that no verification is possible yet.

Other process-algebraic approaches, like $\chi$~\cite{BeekMRRS06}, have been developed for modeling and simulation. Verification is still limited to small fragments that can be translated directly to other verification tools like PHAVer or UPPAAL, which do not support distributed hybrid systems.

Our approach is completely different. It is based on first-order structures and dynamic logic. We focus on developing a logic that supports temporal and non-temporal statements about distributed hybrid dynamics and is amenable to automated theorem proving in the logic itself.

Our previous work and other verification approaches for static hybrid systems cannot verify distributed hybrid systems. Distributed hybrid systems may
have an unbounded and changing number of components/participants, which cannot be represented with any fixed number of dimensions of the state space.

Based on~\cite{Pratt79}, Beckert and Schlager~\cite{BeckertS01} added separate trace modalities to dynamic logic and presented a relatively complete calculus. Their approach only handles discrete state spaces. In contrast, \qdtl works for hybrid programs with continuous and structural/dimensional dynamics.

Davoren and Nerode~\cite{DavorenN00} extended the propositional modal $\mu$-calculus with a semantics in hybrid systems and examine topological aspects. In~\cite{DavorenCM04}, Davoren et al. gave a semantics in general flow systems for a generalisation of $\text{CTL}^{\ast}$~\cite{EmersonH86}. In both cases, the authors of~\cite{DavorenN00} and~\cite{DavorenCM04} provided Hilbert-style calculi to prove formulas that are valid for all systems simultaneously using abstract actions.

The strength of our logic primarily is that it is a first-order dynamic logic: it handles actual hybrid programs rather than only abstract actions of unknown effect. Our calculus directly supports verification of quantified hybrid programs with continuous evolution and structural/dimensional changes. First-order dynamic logic is more expressive and calculi are deductively stronger than other approaches~\cite{BeckertS01,Leivan04}.
\section{Syntax of Quantified Differential Temporal Dynamic Logic}\label{sec:3} 
As a formal logic for verifying temporal specifications of distributed hybrid systems, we introduce {\em quantified differential temporal dynamic logic} (\qdtl). \qdtl extends dynamic logic for reasoning about system runs~\cite{HarelKT00} with many-sorted first-order logic for reasoning about all $(\forall i : \kern -0.17em A \ \phi)$ or some $(\exists i : \kern -0.17em A \ \phi)$ objects of a sort $A$, e.g.,
the sort of all aircraft, and three other concepts:\vspace{.15cm}\\
{\em Quantified hybrid programs}. The behavior of distributed hybrid systems can be described by quantified hybrid programs~\cite{DBLP:conf/csl/Platzer10,DBLP:conf/csl/Platzer10:TR}, which generalize regular programs from dynamic logic~\cite{HarelKT00} to distributed hybrid changes. The distinguish feature of quantified hybrid programs is that they provide uniform discrete transitions, continuous evolutions, and structural/dimension changes along quantified assignments and quantified differential equations, which can be combined by regular control operations.\vspace{.15cm}\\
{\em Modal operators}. Modalities of dynamic logic express statements about all possible behavior $([\alpha]\pi)$ of a system $\alpha$, or about the existence of a trace $(\langle \alpha \rangle \pi)$, satisfying condition $\pi$. Unlike in standard dynamic logic, $\alpha$ is a model of a distributed hybrid system. We use quantified hybrid programs to describe $\alpha$ as in~\cite{DBLP:conf/csl/Platzer10,DBLP:conf/csl/Platzer10:TR}. Yet, unlike in standard dynamic logic~\cite{HarelKT00} or quantified differential dynamic logic (\qdl)~\cite{DBLP:conf/csl/Platzer10,DBLP:conf/csl/Platzer10:TR}, $\pi$ is a {\em trace formula} in \qdtl, and $\pi$ can refer to all states that occur {\em during} a trace using temporal operators.\vspace{.15cm}\\
{\em Temporal operators}. For \qdtl, the temporal trace formula $\Box \phi$ expresses that the formula $\phi$ holds all along a trace selected by $[\alpha]$ or $\langle \alpha \rangle$. For instance, the state formula $\langle \alpha \rangle \Box \phi$ says that the state formula $\phi$ holds at every state along at least one trace of $\alpha$. Dually, the trace formula $\Diamond \phi$ expresses that $\phi$ holds at some point during such a trace. It can occur in a state formula $\langle \alpha \rangle \Diamond \phi$ to express that there is such a state in some trace of $\alpha$, or as $[\alpha]\Diamond\phi$ to say that, along each trace, there is a state satisfying $\phi$. In this paper, the primary focus of attention is on homogeneous combinations of path and trace quantifiers like $[\alpha] \Box \phi$ or $\langle \alpha \rangle \Diamond \phi$.
\subsection{Quantified Hybrid Programs} 
\qdtl supports a (finite) number of object {\em sorts}, e.g., the sort of all aircraft, or the sort of all cars. For
continuous quantities of distributed hybrid systems like positions or velocities,
we add the sort $\mathbb{R}$ for real numbers. {\em Terms} of \qdtl are built from a set of (sorted) function/variable symbols
as in many-sorted first-order logic.
For representing appearance and disappearance of objects while running QHPs, we use an existence function symbol $\mathsf{E}(\cdot)$ that has
value $\mathsf{E}(o) = 1$ if object $o$ exists, and has value $\mathsf{E}(o) = 0$ when object $o$ disappears or has not been created yet.
We use $0, 1, +, -, \cdot$ with the usual notation and fixed semantics for real arithmetic. For
$n \geq 0$ we abbreviate $f(s_1, \ldots, s_n)$ by $f(\boldsymbol{s})$ using vectorial notation and we use
$\boldsymbol{s} = \boldsymbol{t}$ for element-wise equality.

As a system model for distributed hybrid systems, \qdtl uses {\em quantified hybrid programs} (QHP)~\cite{DBLP:conf/csl/Platzer10,DBLP:conf/csl/Platzer10:TR}. The quantified hybrid programs occurring in dynamic modalities of \qdtl are regular programs from dynamic logic~\cite{HarelKT00} to which quantified assignments and quantified differential equation systems for
distributed hybrid dynamics are added. QHPs are defined by the following grammar ($\alpha$, $\beta$ are QHPs, $\theta$ a term, $i$ a variable of sort $A$, $f$ is a function symbol, $\boldsymbol{s}$ is a vector of terms with sorts compatible to $f$, and $\chi$ is a formula of first-order logic):
\[\alpha, \beta ::= \forall i : \kern -0.17em A \ f(\boldsymbol{s}) : = \theta \ | \ \forall i : \kern -0.17em A \ f(\boldsymbol{s})' = \theta \with \chi \ | \ ?\chi \ | \ \alpha \cup \beta \ | \ \alpha;\beta \ | \ \alpha^{\ast}\]
The effect of {\em quantified assignment} $\forall i : \kern -0.3em A \ f(\boldsymbol{s}) : = \theta$ is an instantaneous discrete jump assigning $\theta$ to $f(\boldsymbol{s})$ simultaneously for all objects $i$ of sort $A$. The QHP $\forall i : \kern -0.17em C \ a(i) := a(i) + 1$, for example,
expresses that all cars $i$ of sort $C$ simultaneously increase their acceleration $a(i)$. The effect of {\em quantified differential equation} $ \forall i : \kern -0.17em A \ f(\boldsymbol{s})' = \theta \with \chi$ is a continuous evolution where, for all objects $i$ of sort $A$, all differential equations
$f(\boldsymbol{s})' = \theta$ hold and formula $\chi$ holds throughout the evolution (the state
remains in the region described by $\chi$). The dynamics of QHPs changes the interpretation
of terms over time: $f(\boldsymbol{s})'$ is intended to denote the derivative of the
interpretation of the term $f(\boldsymbol{s})$ over time during continuous evolution, not the
derivative of $f(\boldsymbol{s})$ by its argument $\boldsymbol{s}$. For $f(\boldsymbol{s})'$ to be defined, we assume $f$ is an
$\mathbb{R}$-valued function symbol. For simplicity, we assume that $f$ does not occur in $\boldsymbol{s}$.
In most quantified assignments/differential equations $\boldsymbol{s}$ is just $i$. For instance, the following QHP expresses that
all cars $i$ of sort $C$ drive by $\forall i : \kern -0.17em C \ x(i)'' = a(i)$ such that their position $x(i)$ changes continuously
according to their respective acceleration $a(i)$.

The effect of test $? \chi$ is a {\em skip} (i.e., no change) if formula $\chi$ is
true in the current state and {\em abort} (blocking the system run by a failed assertion),
otherwise. {\em Nondeterministic choice} $\alpha \cup \beta$ is for alternatives in the behavior of
the distributed hybrid system. In the {\em sequential composition} $\alpha ; \beta$, QHP $\beta$ starts
after $\alpha$ finishes ($\beta$ never starts if $\alpha$ continues indefinitely). {\em Nondeterministic
repetition} $\alpha^{\ast}$ repeats $\alpha$ an arbitrary number of times, possibly zero times.

Structural dynamics of distributed hybrid systems corresponds to quantified assignments to function terms and we model the appearance of new participants in the distributed hybrid system, e.g., new aircraft appearing into the local flight scenario, by a program $n := \mathsf{new} \ A$ (see~\cite{DBLP:conf/csl/Platzer10,DBLP:conf/csl/Platzer10:TR} for details).
\subsection{State and Trace Formulas} 
The formulas of \qdtl are defined similarly to first-order dynamic logic plus
many-sorted first-order logic. However, the modalities $[\alpha]$ and $\langle \alpha \rangle$ accept trace formulas that refer to the temporal behavior of {\em all} states along a trace. Inspired by CTL and $\text{CTL}^{\ast}$~\cite{EmersonC82,EmersonH86}, we distinguish between state formulas, which are true or false in states, and trace formulas, which are true or false for system traces.

The {\em state formulas} of \qdtl are defined by the following grammar ($\phi, \psi$ are state formulas, $\pi$ is a trace formula, $\theta_1, \theta_2$ are terms of the same sort, $i$ is a variable of sort $A$, and $\alpha$ is a QHP):
\[\phi, \psi ::= \theta_1 = \theta_2 \ | \ \theta_1 \geq \theta_2 \ | \ \neg \phi \ | \ \phi \wedge \psi \ | \ \forall i : \kern -0.17em A \ \phi \ | \ \exists i : \kern -0.17em A \ \phi \ | \ [\alpha]\pi \ | \ \langle \alpha \rangle \pi\]
We use standard abbreviations to define $\leq, >, <, \vee, \rightarrow$. Sorts $A \not = \mathbb{R}$ have no
ordering and only $\theta_1 = \theta_2$ is allowed. For sort $\mathbb{R}$, we abbreviate $\forall x : \kern -0.17em \mathbb{R} \ \phi$ by $\forall x \phi$.

The {\em trace formulas} of \qdtl are defined by the following grammar ($\pi$ is a trace formula and $\phi$ is a state formula):
\[ \pi ::= \phi \ | \ \Box \phi \ | \ \Diamond \phi \]
Formulas without $\Box$ and $\Diamond$ are {\em non-temporal} \qdl \xspace {\em formulas}. Unlike in CTL, state formulas are true on a trace if they hold for the {\em last} state of a trace, not for the first. Thus, $[\alpha]\phi$ expresses that $\phi$ is true at the end of each trace of $\alpha$. In contrast, $[\alpha]\Box\phi$ expresses that $\phi$ is true all along all states of every trace of $\alpha$. This combination gives a smooth embedding of non-temporal \qdl into \qdtl and makes it possible to define a compositional calculus. Like CTL, \qdtl allows nesting with a branching time semantics~\cite{EmersonC82}, e.g., $[\alpha] \Box ((\forall i : \kern -0.17em C \ x(i) \geq 2) \rightarrow \langle \beta \rangle \Diamond (\forall i : \kern -0.17em C \ x(i) \leq 0))$. In the following, all formulas and terms have to be well-typed. For short notation, we allow conditional terms of the form $\mathsf{if} \kern .17em \phi \kern .17em \mathsf{then} \kern .17em \theta_1 \kern .17em \mathsf{else} \kern .17em \theta_2 \kern .17em \mathsf{fi}$ (where $\theta_1$ and $\theta_2$ have the same sort). This term evaluates to $\theta_1$ if the formula $\phi$ is true and to $\theta_2$ otherwise. We consider formulas with conditional terms as
abbreviations, e.g., $\psi(\mathsf{if}\kern .17em \phi \kern .17em \mathsf{then} \kern .17em \theta_1 \kern .17em \mathsf{else} \kern .17em \theta_2)$ for ($\phi \rightarrow \psi(\theta_1)) \wedge (\neg \phi \rightarrow \psi(\theta_2))$.

\begin{example}
Let $C$ be the sort of all cars. By $x(i)$, we denote the position of car $i$, by $v(i)$ its velocity and by $a(i)$ its acceleration. Then the \qdtl formula
\[ (\forall i : \kern -0.17em C \ x(i) \geq 0) \rightarrow [\forall i : \kern -0.17em C \ x(i)' = v(i), v(i)' = a(i) \with v(i) \geq 0] \Box (\forall i : \kern -0.17em C \ x(i) \geq 0)      \]
says that, if all cars start at a point to the right of the origin and we only allow them to evolve as long as all of them have nonnegative velocity, then they {\em always} stay up to the right of the origin. In this case, the QHP just consists of a quantified differential equation expressing that the position $x(i)$ of car $i$ evolves over time according to the velocity $v(i)$, which evolves according to its acceleration $a(i)$. The constraint $v(i) \geq 0$ expresses that
the cars never move backwards, which otherwise would happen eventually in the case of braking $a(i) < 0$. This formula is indeed valid, and we would be able to use the techniques developed in this paper to prove it.
\end{example}
\section{Semantics of Quantified Differential Temporal Dynamic Logic}\label{sec:4} 
In standard dynamic logic~\cite{HarelKT00} and the logic \qdl~\cite{DBLP:conf/csl/Platzer10,DBLP:conf/csl/Platzer10:TR}, modalities only refer to the final
states of system runs and the semantics is a reachability relation on states:
State $\tau$ is reachable from state $\sigma$ using $\alpha$ if there is a run of $\alpha$ which terminates
in $\tau$ when started in $\sigma$. For \qdtl, however, formulas can refer to intermediate
states of runs as well. Thus, the semantics of a distributed hybrid system $\alpha$ is the set of its
possible {\em traces}, i.e., successions of states that occur during the evolution of $\alpha$.
\subsection{Trace Semantics of Quantified Hybrid Programs} 
A {\em state} $\sigma$ associates an infinite set $\sigma(A)$ of objects with each sort $A$, and
it associates a function $\sigma(f)$ of appropriate type with each function symbol $f$, including $\mathsf{E}(\cdot)$.
For simplicity, $\sigma$ also associates a value $\sigma(i)$ of appropriate type
with each variable $i$. The domain of $\mathbb{R}$ and the interpretation of $0, 1, +, -, \cdot$ is
that of real arithmetic. We assume constant domain for each sort $A$: all states
$\sigma, \tau$ share the same infinite domains $\sigma(A) = \tau(A)$. Sorts $A \not = C$ are disjoint:
$\sigma(A) \cap \sigma(C) = \emptyset$. The set of all states is denoted by $\mathcal{S}$. The state $\sigma_{i}^{e}$ agrees
with $\sigma$ except for the interpretation of variable $i$, which is changed to $e$. In addition, we distinguish a state $\Lambda$ to denote the failure of
a system run when it is {\em aborted} due to a test $?\chi$ that yields {\em false}. In particular, $\Lambda$ can only occur at the end of an aborted system run and marks that there is no further extension.

Distributed hybrid systems evolve along piecewise continuous traces in multi-dimensional space, structural changes, and appearance or disappearance of agents as time passes. Continuous phases are governed by differential equations, whereas discontinuities are caused by discrete jumps. Unlike in discrete cases~\cite{BeckertS01,Pratt79}, traces are not just sequences of states, since distributed hybrid systems pass through uncountably many states even in bounded time. Beyond that, continuous changes are more involved than in pure real-time~\cite{AlurCD90,HenzingerNSY92}, because all variables can evolve along different differential equations. Generalizing the real-time traces of~\cite{HenzingerNSY92}, the following definition captures hybrid behavior by splitting the uncountable succession of states into periods $\nu_i$ that are regulated by the same control law. For discrete jumps, some periods are point flows of duration $0$.

The (trace) semantics of quantified hybrid programs is compositional, that is, the semantics of a complex program is defined as a simple function of the trace semantics of its parts.

\begin{definition}[Hybrid Trace]\label{def:hybridtrace}
A {\em trace} is a (non-empty) finite or infinite sequence $\nu = (\nu_0, \nu_1, \nu_2, \ldots)$ of functions $\nu_k : [0, r_k] \rightarrow \mathcal{S}$ with respective durations $r_k \in \mathbb{R}$ (for $k \in \mathbb{N}$). A {\em position} of $\nu$ is a pair $(k, \zeta)$ with $k \in \mathbb{N}$ and $\zeta$ in the interval $[0, r_k]$; the state of $\nu$ at $(k, \zeta)$ is $\nu_k(\zeta)$. Positions of $\nu$ are ordered lexicographically by
$(k, \zeta) \prec (m, \xi)$ iff either $k < m$, or $k = m$ and $\zeta < \xi$. Further, for a state $\sigma \in \mathcal{S}$, $\hat{\sigma} : 0 \mapsto \sigma$ is the {\em point flow} at $\sigma$ with duration $0$. A trace {\em terminates} if it is a finite sequence $(\nu_0, \nu_1, \ldots, \nu_n)$ and $\nu_n(r_n) \not = \Lambda$. In that case, the last state $\nu_n(r_n)$ is denoted by {\em last}{\kern 0.25em}$\nu$. The first state $\nu_0(0)$ is denoted by {\em first}{\kern 0.25em}$\nu$.
\end{definition}
Unlike in~\cite{AlurCD90,HenzingerNSY92}, the definition of traces also admits finite traces of bounded duration, which is necessary for compositionality of traces in $\alpha;\beta$. The semantics of quantified hybrid programs $\alpha$ as the set $\mu(\alpha)$ of its possible traces depends on valuations $\sigma \llbracket \cdot \rrbracket$ of formulas and terms at intermediate states $\sigma$. The valuation of formulas will be defined in Definition~\ref{def:semanticsofformulas}. Especially, we use $\sigma_{i}^{e}\llbracket \cdot \rrbracket$ to denote the valuations of terms and formulas in state $\sigma_{i}^{e}$, i.e., in state $\sigma$ with $i$ interpreted as $e$.

\begin{definition}[Trace Semantics of Quantified Hybrid Programs]\label{def:tracesemanticsofQHP}
 The {\em trace semantics}, $\mu(\alpha)$, {\em of a quantified hybrid program} $\alpha$, is the set of all its possible hybrid traces and is defined inductively as follows:
\begin{enumerate}
\item $\mu(\forall i : \kern -0.17em A \ f(\boldsymbol{s}) := \theta) = \{ (\hat{\sigma}, \hat{\tau}) : \sigma \in \mathcal{S}$ and state $\tau$ is identical to $\sigma$ except that at each position $\boldsymbol{o}$ of $f$: if $\sigma_i^{e} \llbracket \boldsymbol{s} \rrbracket = \boldsymbol{o}$ for some object $e \in \sigma(A)$, then $\tau(f)(\sigma_i^e\llbracket \boldsymbol{s} \rrbracket) = \sigma_i^e \llbracket \theta \rrbracket.\}$
 \item $\mu(\forall i : \kern -0.17em A \ f(\boldsymbol{s})' = \theta \with \chi) = \{ (\varphi) : 0 \leq r \in \mathbb{R} \ \text{and } \varphi: [0, r] \rightarrow \mathcal{S}$ is a function satisfying the following conditions. At each time $t \in [0,r]$, state $\varphi(t)$ is identical to $\varphi(0)$, except that at each position $\boldsymbol{o} \ \text{of} \ f: \ \text{if} \ \sigma_i^{e} \llbracket \boldsymbol{s} \rrbracket = \boldsymbol{o}$ for some object $e \in \sigma(A)$, then, at each time $\zeta \in [0,r]$:
     \begin{itemize}
     \item The differential equations hold and derivatives exist (trivial for $r = 0$):
     \[\inferencerule{\mathsf{d} ({\varphi(t)}^{e}_{i}\llbracket f(\boldsymbol{s}) \rrbracket)}{\mathsf{d} t}(\zeta) = ({\varphi(\zeta)}^{e}_{i}\llbracket \theta \rrbracket) \]
     \item The evolution domains is respected: ${\varphi(\zeta)}^{e}_{i} \llbracket \chi \rrbracket = \text{true}.\}$
     \end{itemize}
 \item $\mu(?\chi) = \{ (\hat{\sigma}) : \sigma \llbracket \chi \rrbracket = \text{true} \} \cup \{ (\hat{\sigma}, \hat{\Lambda}): \sigma \llbracket \chi \rrbracket = \text{false}\}$
 \item $\mu(\alpha \cup \beta) = \mu(\alpha) \cup \mu(\beta)$
 \item $\mu(\alpha;\beta) = \{ \nu \circ \varsigma : \nu \in \mu(\alpha), \varsigma \in \mu(\beta) \ \text{when} \ \nu \circ \varsigma \ \text{is defined}\}$; the composition of $\nu = (\nu_0, \nu_1, \nu_2, \ldots)$ and $\varsigma = (\varsigma_0, \varsigma_1, \varsigma_2, \ldots)$ is
     \begin{equation*}
\nu \circ \varsigma = \left\{
\begin{array}{ll}
(\nu_0, \ldots, \nu_n, \varsigma_0, \varsigma_1, \ldots) & \text{if } \nu \ \text{terminates at } \nu_n \ \text{and } \emph{\text{last}} {\kern 0.25em}\nu = \emph{\text{first}}{\kern 0.25em}\varsigma\\
\nu & \text{if } \nu \ \text{does not terminate}\\
\text{not defined} & \text{otherwise}
\end{array} \right.
\end{equation*}
 \item $\mu(\alpha^{\ast}) = \bigcup_{n \in \mathbb{N}} \mu(\alpha^{n})$, where $\alpha^{n+1} := (\alpha^n;\alpha)$ for $n \geq 1$, as well as $\alpha^1 := \alpha$ and $\alpha^{0} := (?\text{true})$.
\end{enumerate}
\end{definition}
\subsection{Valuation of State and Trace Formulas} 
\begin{definition}[Valuation of Formulas]\label{def:semanticsofformulas}
 The valuation of state and trace formulas is defined respectively. For state formulas, the {\em valuation} $\sigma \llbracket \cdot \rrbracket$ with respect to state $\sigma$ is defined as follows:
\begin{enumerate}
\item $\sigma \llbracket \theta_1 = \theta_2 \rrbracket = $ true  iff $\sigma \llbracket \theta_1 \rrbracket = \sigma \llbracket \theta_2 \rrbracket$; accordingly for $\geq$.
\item $\sigma \llbracket \phi \wedge \psi \rrbracket =$ true iff $\sigma \llbracket \phi \rrbracket =$ true and $\sigma \llbracket \psi \rrbracket =$ true; accordingly for $\neg$.
\item $\sigma \llbracket \forall i : \kern -0.17em A \ \phi \rrbracket =$ true iff $\sigma_i^e \llbracket \phi \rrbracket =$ true for all objects $e \in \sigma(A)$.
\item $\sigma \llbracket \exists i : \kern -0.17em A \ \phi \rrbracket =$ true iff $\sigma_i^e \llbracket \phi \rrbracket =$ true for some object $e \in \sigma(A)$.
\item $\sigma \llbracket [\alpha] \pi \rrbracket =$ true iff for each trace $\nu \in \mu(\alpha)$ that starts in {\em first}{\kern 0.25em}$\nu = \sigma$, if $\nu \llbracket \pi \rrbracket$ is defined, then $\nu \llbracket \pi \rrbracket =$ true.
\item $\sigma \llbracket \langle \alpha \rangle \pi \rrbracket =$ true iff there is a trace $\nu \in \mu(\alpha)$ starting in {\em first}{\kern 0.25em}$\nu = \sigma$ such that $\nu \llbracket \pi \rrbracket$ is defined and $\nu \llbracket \pi \rrbracket =$ true.

\end{enumerate}

For trace formulas, the {\em valuation} $\nu \llbracket \cdot \rrbracket$ with respect to trace $\nu$ is defined as follows:
\begin{enumerate}
\item If $\phi$ is a state formula, then $\nu \llbracket \phi \rrbracket =$ {\em last}{\kern 0.25em}$\nu \llbracket \phi \rrbracket$ if $\nu$ terminates, whereas $\nu \llbracket \phi \rrbracket$ is not defined if $\nu$ does not terminate.
\item $\nu \llbracket \Box \phi \rrbracket =$ true iff $\nu_{k}(\zeta) \llbracket \phi \rrbracket =$ true for all positions $(k, \zeta)$ of $\nu$ with $\nu_k(\zeta) \not = \Lambda$.
\item $\nu \llbracket \Diamond \phi \rrbracket =$ true iff $\nu_{k}(\zeta) \llbracket \phi \rrbracket =$ true for some position $(k, \zeta)$ of $\nu$ with $\nu_k(\zeta) \not = \Lambda$.

\end{enumerate}
\end{definition}
As usual, a (state) formula is {\em valid} if it is true in all states. Further for (state) formula $\phi$ and state $\sigma$ we write $\sigma \models \phi$ iff $\sigma \llbracket \phi \rrbracket$ = {\em true}. We write $\sigma \not \models \phi$ iff $\sigma \llbracket \phi \rrbracket$ = {\em false}. Likewise, for trace formula $\pi$ and trace $\nu$ we write $\nu \models \pi$ iff $\nu \llbracket \pi \rrbracket$ = {\em true} and $\nu \not \models \pi$ iff $\nu \llbracket \pi \rrbracket$ = {\em false}. In particular, we only write $\nu \models \pi$ or $\nu \not \models \pi$ if $\nu \llbracket \pi \rrbracket$ is defined, which it is not the case if $\pi$ is a state formula and $\nu$ does not terminate.
\subsection{Conservative Temporal Extension}
 The following result shows that the extension of \qdtl by temporal operators does not change the meaning of non-temporal \qdl formulas. The trace semantics given in Definition~\ref{def:semanticsofformulas} is equivalent to the final state reachability relation semantics~\cite{DBLP:conf/csl/Platzer10,DBLP:conf/csl/Platzer10:TR} for the sublogic \qdl of \qdtl.

\begin{proposition}\label{prop:cextension}
The logic {\em \qdtl} is a {\em conservative extension} of non-temporal {\em \qdl},
i.e., the set of valid {\em \qdl} formulas is the same with respect to transition reachability
semantics of {\em \qdl}~\cite{DBLP:conf/csl/Platzer10,DBLP:conf/csl/Platzer10:TR} as with respect to the trace semantics of {\em \qdtl} (Definition~\ref{def:semanticsofformulas}).
\end{proposition}
\section{Safety Properties in Distributed Air Traffic Control}\label{sec:5} 
In air traffic control, collision avoidance maneuvers~\cite{DowekMC05,TomlinPS98} are used to resolve conflicting flight paths that arise during free flight. We
consider the roundabout collision avoidance maneuver for air traffic control~\cite{TomlinPS98}. In the literature, formal verification of the hybrid dynamics of air traffic control focused on a fixed number of aircraft, usually two. In reality, many more aircraft are in the same flight corridor, even if not all of them participate in the same maneuver. They may be involved in multiple distributed maneuvers at the same time, however. Perfect global trajectory planning quickly becomes infeasible then. The verification itself also becomes much more complicated for three aircraft already. Explicit replication of the system dynamics $n$ times is computationally infeasible for larger $n$. Yet, collision avoidance maneuvers need to work for an (essentially) unbounded number of aircraft. Because global trajectory planning is infeasible, the appearance of other aircraft into a local collision avoidance maneuver always has to be expected and managed safely. See Fig.~\ref{fig:roundabout} for a general illustration of roundabout-style collision avoidance maneuvers and the phenomenon of dynamic appearance of some new aircraft $z$ into the horizon of relevance.
\begin{wrapfigure}{r}{0.30\textwidth}
\vspace{-30pt}
  \begin{center}
    \includegraphics[width=0.30\textwidth]{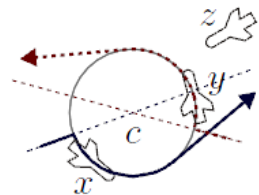}
  \end{center}
  \vspace{-.5cm}
  \caption{\scriptsize{Roundabout collision avoidance maneuver with new appearance}}
  \label{fig:roundabout}
  \vspace{-10pt}
\end{wrapfigure}
\quad The resulting flight control system has several characteristics of hybrid dynamics. But it is not a hybrid system and does not even have a fixed finite number of variables in a fixed finite-dimensional state space. The system forms a distributed hybrid system, in which all aircraft fly at the same time and new aircraft may appear from remote areas into the local flight scenario. Let $A$ be the sort of all aircraft. Each aircraft $i$ has a position $x(i) = (x_1(i), x_2(i))$ and a velocity vector $d(i) = (d_1(i), d_2(i))$. We model the continuous dynamics of an aircraft $i$ that follows a flight curve with an angular velocity $\omega(i)$ by the (function) differential equation:
\begin{equation}
{x_1(i)}' = d_1(i), {x_2(i)}' = d_2(i), {d_1(i)}' = - \omega(i)d_2(i), {d_2(i)}' = \omega(i)d_1(i) \tag{${\mathcal{F}}_{\omega(i)}(i)$}
\end{equation}
This differential equation, which we denote by ${\mathcal{F}}_{\omega(i)}(i)$, is the standard equation for curved flight from the literature~\cite{TomlinPS98}, but lifted to function symbols that are parameterized by aircraft $i$. Now the quantified differential equation $\forall i : \kern -0.17em A \ {\mathcal{F}}_{\omega(i)}(i)$ characterizes that {\em all} aircraft $i$ fly along their respective (function) differential equation ${\mathcal{F}}_{\omega(i)}(i)$ according to their respective angular velocities $\omega(i)$ at the same time. This quantified differential equation captures what no finite-dimensional differential equation system could ever do. It characterizes the simultaneous movement of an unbounded, arbitrary, and even growing or shrinking set of aircraft.

Two aircraft $i$ and $j$ have violated the safe separation property if they falsify the following formula
\[\mathcal{P}(i,j) \equiv i = j \vee (x_1(i) - x_1(j))^{2} + (x_2(i) - x_2(j))^{2} \geq p^2 \]
which says that aircraft $i$ and $j$ are either identical or separated by at least the protected zone $p$ (usually 5mi). For the aircraft control system to be safe, all aircraft have to be safely separated, i.e., need to satisfy $\forall i, j : \kern -0.17em A \ \mathcal{P}(i,j)$. It is crucial that this formula holds at {\em every} point in time during the system evolution, not only at its beginning or at its end. Hence, we need to consider temporal safety properties. For instance, \qdtl can analyze the following temporal safety properties of a part of the distributed roundabout collision avoidance maneuver for air traffic control:
\begin{equation}
\forall i, j : \kern -0.17em A \ \mathcal{P}(i,j) \wedge \forall i, j : \kern -0.17em A \ \mathcal{T}(i,j) \rightarrow [\forall i : \kern -0.17em A \  {\mathcal{F}}_{\omega(i)}(i)] \Box \kern 0.08em \forall i, j : \kern -0.17em A \ \mathcal{P}(i,j) \label{eq:safety}
\end{equation}
\begin{equation}
\begin{array}{l}
\kern 1.1cm \forall i, j : \kern -0.17em A \ \mathcal{P}(i,j) \wedge \forall i, j : \kern -0.17em A \ \mathcal{T}(i,j) \rightarrow \\
 \lbrack \forall i : \kern -0.17em A \ t := 0; \forall i : \kern -0.17em A \  {\mathcal{F}}_{\omega(i)}(i), t' = 1 \with \forall i : \kern -0.17em A \ t \leq T; ?(\forall i : \kern -0.17em A \ t = T)] \Box \kern 0.08em \forall i, j : \kern -0.17em A \ \mathcal{P}(i,j) \label{eq:safetyproperty}
\end{array}
\end{equation}
where $\mathcal{T}(i,j) \equiv d_1(i) - d_1(j) = - \omega(x_2(i) - x_2(j)) \wedge d_2(i) - d_2(j) = \omega(x_1(i) - x_1(j))$, $t$ is a clock variable, and $T$ is some bounded time.

The temporal safety invariant in (\ref{eq:safety}) expresses that the circle phase of roundabout maneuver {\em always} stays collision-free indefinitely for an arbitrary number of aircraft. That is the most crucial part because we have to know the aircraft {\em always} remain safe during the actual roundabout collision avoidance circle. The condition $\forall i, j : \kern -0.17em A \ \mathcal{T}(i,j)$ characterizes compatible tangential maneuvering choices. Without a condition like $\mathcal{T}(i,j)$, roundabouts can be unsafe~\cite{APlatzer10}. For a systematic derivation of how to construct $\mathcal{T}(i,j)$, we refer to the work~\cite{APlatzer10}. As a variation of (\ref{eq:safety}), the temporal safety property in (\ref{eq:safetyproperty}) states that, for an arbitrary number of aircraft, the circle procedure of roundabout maneuver cannot produce collisions at any point in its bounded duration $T$. This variation restricts the continuous evolution to take exactly $T$ time units (the evolution domain region restricts the evolution to $t \leq T$ and the subsequent test to $?(\forall i : \kern -0.17em A \ t = T)$) and no intermediate state is visible as a final state anymore. Thus, the temporal modality $\Box$ in (\ref{eq:safetyproperty}) is truly necessary. We will use the techniques developed in this paper to verify these temporal safety properties in the distributed roundabout flight collision avoidance maneuver.
\section{Proof Calculus for Temporal Properties}\label{sec:6} 
In this section, we introduce a sequent calculus for verifying temporal specifications of distributed hybrid systems in \qdtl. With the basic idea being to perform a symbolic decomposition, the calculus transforms quantified hybrid programs successively into simpler logical formulas describing their effects. Statements about the temporal behavior of a quantified hybrid program are successively reduced to corresponding non-temporal statements about the intermediate states.
\subsection{Proof Rules} 
In Fig.~\ref{fig:rules}, we present a proof calculus for \qdtl that inherits the proof rules of \qdl from~\cite{DBLP:conf/csl/Platzer10,DBLP:conf/csl/Platzer10:TR} and adds new proof rules for temporal modalities. We use the sequent notation informally for a systematic proof structure. A {\em sequent} is of the form $\Gamma \rightarrow \Delta$, where the {\em antecedent} $\Gamma$ and {\em succedent} $\Delta$ are finite sets of formulas. The semantics of $\Gamma \rightarrow \Delta$ is that of the formula $\bigwedge_{\phi \in \Gamma} \phi \rightarrow \bigvee_{\psi \in \Delta} \psi$ and will be treated as an abbreviation. As usual in sequent calculus, the proof rules are applied backwards from the {\em conclusion} (goal below horizontal bar) to the {\em premise}s (subgoals above bar).\vspace{-.3cm}
\begin{figure}[!ht]
\hspace*{-0.32in}
\begin{minipage}{1.0\linewidth}
\setcounter{mpfootnote}{\value{footnote}}
\renewcommand{\thempfootnote}{\arabic{mpfootnote}}
\begin{tabular}{l l l l l l}
{\scriptsize ({\em ax}) $\inferencerule{}{\phi \rightarrow \phi}$} &
{\scriptsize ($\neg$r) $\inferencerule{\phi \rightarrow }{ \rightarrow \neg \phi}$} &
{\scriptsize ($\neg$l) $\inferencerule{\rightarrow \phi}{ \neg \phi \rightarrow}$} &
{\scriptsize ($\wedge$r) $\inferencerule{\rightarrow \phi \quad \rightarrow \psi}{\rightarrow \phi \wedge \psi}$}  &
{\scriptsize ($\wedge$l) $\inferencerule{\phi, \psi \rightarrow}{\phi \wedge \psi \rightarrow}$} &
{\scriptsize ({\em cut}) $\inferencerule{\rightarrow \phi \quad \phi \rightarrow}{ \rightarrow}$}\\
\vspace{-2.6mm} \\
{\scriptsize ($[;]$) $\inferencerule{[\alpha][\beta]\phi}{[\alpha;\beta] \phi}$} &
{\scriptsize ($\langle ; \rangle$) $\inferencerule{\langle \alpha \rangle \langle \beta \rangle \phi}{\langle \alpha ; \beta \rangle \phi}$} &
{\scriptsize ($[\cup]$) $\inferencerule{[\alpha]\phi \wedge [\beta]\phi}{[\alpha \cup \beta] \phi}$} &
{\scriptsize ($\langle \cup \rangle$) $\inferencerule{\langle \alpha \rangle \phi \vee \langle \beta \rangle \phi}{\langle \alpha \cup \beta \rangle \phi}$} &
{\scriptsize ($[?]$) $\inferencerule{ \chi \rightarrow \phi}{[?\chi]\phi}$} &
{\scriptsize ($\langle ? \rangle$) $\inferencerule{ \chi \wedge \phi}{\langle ?\chi \rangle \phi}$}
\end{tabular}
\begin{tabular}{l}
{\scriptsize ($[']$) $\inferencerule{\forall t \geq 0 ((\forall 0 \leq \tilde{t} \leq t [\forall i : \kern -0.17em A \ \mathcal{S}(\tilde{t})] \chi) \rightarrow [ \forall i : \kern -0.17em A \ \mathcal{S}(t)]\phi)}{[\forall i : \kern -0.17em A \ f(\boldsymbol{s})' = \theta \with \chi] \phi}$ \footnote[1]{\scriptsize{$t, \tilde{t}$ are new variables, $\forall i : \kern -0.17em A \ \mathcal{S}(t)$ is the quantified assignment $\forall i : \kern -0.17em A \ f(\boldsymbol{s}) := y_{\boldsymbol{s}}(t)$ with solutions $y_{\boldsymbol{s}}(t)$ of the (injective) differential equations and $f(\boldsymbol{s})$ as initial values. See~\cite{DBLP:conf/csl/Platzer10,DBLP:conf/csl/Platzer10:TR} for the definition of a {\em injective} quantified assignment or quantified differential equation.}}}  \quad \quad
{\scriptsize ($\langle'\rangle$) $\inferencerule{\exists t \geq 0 ((\forall 0 \leq \tilde{t} \leq t \langle \forall i : \kern -0.17em A \ \mathcal{S}(\tilde{t})\rangle \chi) \wedge \langle \forall i : \kern -0.17em A \ \mathcal{S}(t)\rangle\phi)}{\langle \forall i : \kern -0.17em A \ f(\boldsymbol{s})' = \theta \with \chi \rangle \phi}$\footnotemark[1]} \\
\vspace{-2.6mm}\\
{\scriptsize ($[:=]$) $\inferencerule{\mathsf{if} \ \exists i : \kern -0.17em A \ \boldsymbol{s} = [\mathcal{A}] \boldsymbol{u} \ \mathsf{then} \ \forall i : \kern -0.17em A \ (\boldsymbol{s} = [\mathcal{A}] \boldsymbol{u} \rightarrow \phi(\theta)) \ \mathsf{else} \ \phi(f([\mathcal{A}] \boldsymbol{u})) \ \mathsf{fi}}{\phi([\forall i : A\ f(\boldsymbol{s}) := \theta]f(\boldsymbol{u}))}$\footnote[2]{\scriptsize{Occurrence $f(\boldsymbol{u})$ in $\phi(f(\boldsymbol{u}))$ is not in scope of a modality (admissible substitution) and we abbreviate assignment $\forall i : \kern -0.17em A \ f(\boldsymbol{s}) := \theta$ by $\mathcal{A}$, which is assumed to be injective.}}} \\
\vspace{-2.6mm}\\
{\scriptsize ($\langle := \rangle$) $\inferencerule{\mathsf{if} \ \exists i : \kern -0.17em A \ \boldsymbol{s} = \langle \mathcal{A}\rangle \boldsymbol{u} \ \mathsf{then} \ \exists i : \kern -0.17em A \ (\boldsymbol{s} = \langle \mathcal{A}\rangle \boldsymbol{u} \wedge \phi(\theta)) \ \mathsf{else} \ \phi(f(\langle \mathcal{A}\rangle \boldsymbol{u})) \ \mathsf{fi}}{\phi(\langle \forall i : \kern -0.17em A \ f(\boldsymbol{s}) := \theta \rangle f(\boldsymbol{u}))}$\footnotemark[2]}
\end{tabular}
\begin{tabular}{l l l l}
{\scriptsize ({\em skip}) $\inferencerule{\Upsilon([\forall i : \kern -0.17em A \ f(\boldsymbol{s}) := \theta]\boldsymbol{u})}{[\forall i : \kern -0.17em A \ f(\boldsymbol{s}) := \theta] \Upsilon(\boldsymbol{u})}$\footnote[3]{\scriptsize{$f \not = \Upsilon$ and the quantified assignment $\forall i : \kern -0.17em A \ f(\boldsymbol{s}) := \theta$ is injective. The same rule applies for $\langle \forall i : \kern -0.17em A \ f(\boldsymbol{s}) := \theta \rangle$ instead of $[\forall i : \kern -0.17em A \ f(\boldsymbol{s}) := \theta]$.}}} &
{\scriptsize ($[:\kern -0.3em \ast]$) $\inferencerule{\forall j : \kern -0.17em A \ \phi(\theta)}{[\forall j : \kern -0.17em A \ n := \theta]\phi(n)}$} &
{\scriptsize ($\langle :\kern -0.3em \ast \rangle$) $\inferencerule{\exists j : \kern -0.17em A \ \phi(\theta)}{\langle \forall j : \kern -0.17em A \ n := \theta \rangle \phi(n)}$} &
{\scriptsize ({\em ex}) $\inferencerule{\emph{\text{true}}}{\exists n : \kern -0.17em A \ \mathsf{E}(n) = 0}$ } \\
\vspace{-2.6mm} \\
{\scriptsize ($\forall$r) $\inferencerule{\Gamma \rightarrow \phi(f(X_1, \ldots, X_n)), \Delta}{\Gamma \rightarrow \forall x \phi(x), \Delta}$\footnote[4]{\scriptsize{$f$ is a new (Skolem) function and $X_1, \ldots, X_n$ are all free logical variables of $\forall x \kern 0.1em \phi(x)$.}}} &
{\scriptsize ($\exists$r) $\inferencerule{\Gamma \rightarrow \phi(\theta), \exists x \phi(x), \Delta}{\Gamma \rightarrow \exists x \phi(x), \Delta}$\footnote[5]{\scriptsize{$\theta$ is an abbreviate term, often a new logical variable.}}} &
{\scriptsize ($\forall$l) $\inferencerule{\Gamma, \phi(\theta), \forall x \phi(x) \rightarrow \Delta}{\Gamma, \forall x \phi(x) \rightarrow \Delta}$\footnotemark[5]}&
{\scriptsize ($\exists$l) $\inferencerule{\Gamma, \phi(f(X_1, \ldots, X_n)) \rightarrow \Delta}{\Gamma, \exists x \phi(x) \rightarrow \Delta}$\footnotemark[4]}
\end{tabular}
\begin{tabular}{l l}
{\scriptsize (i$\forall$) $\inferencerule{\text{QE}(\forall X, Y (\mathsf{if} \kern 0.17 em \boldsymbol{s} = \boldsymbol{t} \kern 0.17em \mathsf{then} \kern 0.17em \Phi(X) \rightarrow \Psi(X) \kern 0.17em \mathsf{else} \kern 0.17em \Phi(X) \rightarrow \Psi(Y) \kern 0.17em \mathsf{fi}))}{\Phi(f(\boldsymbol{s})) \rightarrow \Psi(f(\boldsymbol{t}))}$\footnote[6]{\scriptsize{$X, Y$ are new variables of sort $\mathbb{R}$. QE needs to be applicable in the premises.}}} &
{\scriptsize (i$\exists$) $\inferencerule{\text{QE}(\exists X \bigwedge_{i}(\Phi_i \rightarrow \Psi_i))}{\Phi_1 \rightarrow \Psi_1 \ \ldots \ \Phi_n \rightarrow \Psi_n}$\footnote[7]{\scriptsize{    Among all open branches, the free (existential) logical variable $X$ of sort $\mathbb{R}$ only occurs in the branch $\Phi_i \rightarrow \Psi_i$.  QE needs to be defined for the formula in the premises, especially, no Skolem dependencies on $X$ occur.                              }}}
\end{tabular}
\begin{tabular}{l l l l}
{\scriptsize ($[]${\em gen}) $\inferencerule{\phi \rightarrow \psi}{\Gamma, [\alpha]\phi \rightarrow [\alpha] \psi, \Delta}$} &
{\scriptsize ($\langle \rangle${\em gen}) $\inferencerule{\phi \rightarrow \psi}{\Gamma, \langle \alpha \rangle \phi \rightarrow \langle \alpha \rangle \psi, \Delta}$} &
{\scriptsize ({\em ind}) $\inferencerule{\phi \rightarrow [\alpha] \phi}{\Gamma, \phi \rightarrow [\alpha^{\ast}]\phi, \Delta}$} &
{\scriptsize ({\em con}) $\inferencerule{v > 0 \wedge \varphi(v) \rightarrow \langle \alpha \rangle \varphi(v - 1)}{\Gamma, \exists v \kern 0.17em \varphi(v) \rightarrow \langle \alpha^{\ast} \rangle \exists v \leq 0 \kern 0.11em \varphi(v), \Delta}$\footnote[8]{\scriptsize{        logical variable $v$ does not occur in $\alpha$.            }}}
\end{tabular}
\begin{tabular}{l l}
{\scriptsize ({\em DI}) $\inferencerule{\chi \rightarrow [\forall i : \kern -0.17em A \ f(\boldsymbol{s})' := \theta] D(\phi)}{\phi \rightarrow [\forall i : \kern -0.17em A \ f(\boldsymbol{s})' = \theta \with \chi]\phi}$\footnote[9]{\scriptsize{The operator $D$, as defined in~\cite{DBLP:conf/csl/Platzer10,DBLP:conf/csl/Platzer10:TR}, is used to computer syntactic total derivations of formulas algebraically.}}} &
{\scriptsize ({\em DC}) $\inferencerule{\phi \rightarrow [\forall i : \kern -0.17em A \ f(\boldsymbol{s})' = \theta \with \chi]\psi \quad \phi \rightarrow [\forall i : \kern -0.17em A \ f(\boldsymbol{s})' = \theta \with \chi \wedge \psi]\phi}{\phi \rightarrow [\forall i : \kern -0.17em A \ f(\boldsymbol{s})' = \theta \with \chi]\phi}$}
\end{tabular}
\begin{tabular}{l l l l}
{\scriptsize ($[\cup]${\tiny $\Box$}) $\inferencerule{[\alpha] \pi \wedge [\beta] \pi}{[\alpha \cup \beta]\pi}$\footnote[10]{\scriptsize{  $\pi$ is a trace formula, whereas $\phi$ and $\psi$ are (state) formulas. Unlike $\phi$ and $\psi$, the trace formula $\pi$ may thus begin with a temporal modality $\Box$ or $\Diamond$.       }}} &
{\scriptsize ($\langle \cup \rangle \diamond$) $\inferencerule{\langle \alpha \rangle \pi \vee \langle \beta \rangle \pi}{\langle \alpha \cup \beta \rangle \pi}$\footnotemark[10]} &
{\scriptsize ($[;]${\tiny $\Box$}) $\inferencerule{[\alpha]\Box\phi \wedge [\alpha][\beta]\Box\phi}{[\alpha; \beta]\Box \phi}$}  &
{\scriptsize ($\langle ; \rangle \diamond$) $\inferencerule{\langle \alpha \rangle \Diamond \phi \vee \langle \alpha \rangle \langle \beta \rangle \Diamond \phi}{\langle \alpha; \beta \rangle \Diamond \phi}$} \\
{\scriptsize ($[?]${\tiny $\Box$}) $\inferencerule{\phi}{[?\chi] \Box \phi}$} &
{\scriptsize ($\langle ? \rangle \diamond$) $\inferencerule{\phi}{\langle ? \chi \rangle \Diamond \phi}$} &
   \quad &
   \quad \\
{\scriptsize ($[:=]${\tiny $\Box$}) $\inferencerule{\phi \wedge [\forall i : \kern -0.17em A \ f(\boldsymbol{s}) := \theta] \phi}{[\forall i : \kern -0.17em A \ f(\boldsymbol{s}) := \theta] \Box \phi}$} &
{\scriptsize ($\langle := \rangle \diamond$) $\inferencerule{\phi \vee \langle \forall i : \kern -0.17em A \ f(\boldsymbol{s}) := \theta \rangle \phi}{\langle \forall i : \kern -0.17em A \ f(\boldsymbol{s}) := \theta \rangle \Diamond \phi}$}  &
\quad &
\quad \\
{\scriptsize ($[']${\tiny $\Box$}) $\inferencerule{[\forall i : \kern -0.17em A \ f(\boldsymbol{s})' = \theta \with \chi] \phi}{[\forall i : \kern -0.17em A \ f(\boldsymbol{s})' = \theta \with \chi] \Box \phi}$ } &
{\scriptsize ($\langle ' \rangle \diamond$) $\inferencerule{\langle \forall i : \kern -0.17em A \ f(\boldsymbol{s})' = \theta \with \chi \rangle \phi}{\langle \forall i : \kern -0.17em A \ f(\boldsymbol{s})' = \theta \with \chi \rangle \Diamond \phi}$} &
\quad &
\quad \\
{\scriptsize ($[{}^{\ast n}]${\tiny $\Box$}) $\inferencerule{[\alpha; \alpha^{\ast}] \Box \phi}{[\alpha^{\ast}]\Box \phi}$}  &
{\scriptsize ($\langle {}^{\ast n} \rangle \diamond$) $\inferencerule{\langle \alpha; \alpha^{\ast}\rangle \Diamond \phi}{\langle \alpha^{\ast}\rangle \Diamond \phi}$ } &
{\scriptsize ($[{}^{\ast}]${\tiny $\Box$}) $\inferencerule{[\alpha^\ast][\alpha] \Box \phi}{[\alpha^{\ast}]\Box \phi}$} &
{\scriptsize ($\langle{}^{\ast}\rangle\diamond$) $\inferencerule{\langle \alpha^\ast \rangle\langle\alpha\rangle \Diamond \phi}{\langle \alpha^{\ast}\rangle \Diamond \phi}$}
\end{tabular}
\setcounter{footnote}{\value{mpfootnote}}
\end{minipage}
\caption{\scriptsize{Rule schemata of the proof calculus for quantified differential temporal dynamic logic}}
\label{fig:rules}
\end{figure}
\subsubsection{Inherited Non-temporal Rules}
The \qdtl calculus inherits the (non-temporal) \qdl proof rules. For propositional logic, standard rules {\em ax}--{\em cut} are listed in Fig.~\ref{fig:rules}. Rules $[;]$--$\langle ? \rangle$ work similar to those in~\cite{HarelKT00}. Rules $['], \langle ' \rangle$ handle continuous evolutions
for quantified differential equations with first-order definable solutions. Rules $[:=]$--$\langle :\kern -0.3em \ast \rangle$ handle discrete changes for quantified assignments. Axiom {\em ex} expresses that, for sort $A \not = \mathbb{R}$, there always is a new object $n$ that has not been
created yet ($\mathsf{E}(n) = 0$), because domains are infinite. The quantifier rules $\forall$r--i$\exists$ combine quantifier handling of many-sorted logic based on instantiation with theory reasoning by {\em quantifier elimination} (QE) for the theory of reals. The global rules $[]${\em gen}, $\langle \rangle${\em gen} are G\"{o}del generalization rules and {\em ind} is an induction schema for loops with {\em inductive invariant} $\phi$~\cite{HarelKT00}. Similarly, {\em con} generalizes Harel's convergence rule~\cite{HarelKT00} to the hybrid case with decreasing variant $\varphi$~\cite{APlatzer08}. {\em DI} and {\em DC} are rules for quantified differential equations with {\em quantified differential invariants}~\cite{DBLP:conf/csl/Platzer10,DBLP:conf/csl/Platzer10:TR}. Notice that $[\cup], \langle \cup \rangle$ can be generalized to apply to formulas of the form $[\alpha \cup \beta]\pi$ where $\pi$ is an arbitrary trace formula, and not just a state formula as in \qdl. Thus, $\pi$ may begin with $\Box$ and $\Diamond$, which is why the rules are repeated in this generalized form as $[\cup]${\tiny{$\Box$}} and $\langle \cup \rangle \diamond$ in Fig.~\ref{fig:rules}.\vspace{-.3cm}
\vspace{-.1cm}
\subsubsection{Temporal Rules} The new temporal rules in Fig.~\ref{fig:rules} for the \qdtl calculus successively transform temporal specifications of quantified hybrid programs into non-temporal \qdl formulas. The idea underlying this transformation is to decompose quantified hybrid programs and recursively augment intermediate state transitions with appropriate specifications.

Rule $[;]${\tiny $\Box$} decomposes invariants of $\alpha; \beta$ (i.e., $[\alpha; \beta] \Box \phi$ holds) into an invariant of $\alpha$ (i.e., $[\alpha]\Box\phi$) and an invariant of $\beta$ that holds when $\beta$ is started in {\em any} final state of $\alpha$ (i.e., $[\alpha]([\beta]\Box \phi)$). Its difference with the \qdl rule $[;]$ thus is that the \qdtl rule $[;]${\tiny $\Box$} also checks safety invariant $\phi$ at the symbolic states in between the execution of $\alpha$ and $\beta$, and recursively so because of the temporal modality $\Box$. Rule $[:=]${\tiny $\Box$} expresses that invariants of quantified assignments need to hold before and after the discrete change (similarly for $[?]${\tiny $\Box$}, except that tests do not lead to a state change, so $\phi$ holding before the test is all there is to it). Rule $[']${\tiny $\Box$} can directly reduce invariants of continuous evolutions to non-temporal formulas as restrictions of solutions of quantified differential equations are themselves solutions of different duration and thus already included in the continuous evolutions of $\forall i : \kern -0.17em A \ f(\boldsymbol{s})' = \theta$. The (optional) iteration rule $[{}^{\ast n}]${\tiny $\Box$} can partially unwind loops. It relies on rule $[;]${\tiny $\Box$}.
The dual rules $\langle \cup \rangle \diamond$,$\langle ; \rangle \diamond$,$\langle := \rangle \diamond$,$\langle ? \rangle \diamond$,$\langle ' \rangle \diamond$,$\langle {}^{\ast n} \rangle \diamond$ work similarly. Rules for handling $[\alpha]\Diamond \phi$ and $\langle \alpha \rangle \Box \phi$ are discussed in Section~\ref{sec:8}.

The inductive definition rules $[{}^{\ast}]${\tiny $\Box$} and $\langle{}^{\ast}\rangle\diamond$ completely reduce temporal properties of loops to \qdtl properties of standard non-temporal \qdl modalities such that standard induction ({\em ind}) or convergence ({\em con}) rules, as listed in Fig.~\ref{fig:rules}, can be used for the outer non-temporal modality of the loop. Hence, after applying the inductive loop definition rules $[{}^{\ast}]${\tiny $\Box$} and $\langle{}^{\ast}\rangle\diamond$, the standard \qdl loop invariant and variant rules can be used for verifying temporal properties of loops without change, except that the postcondition contains temporal modalities.


\subsection{Soundness and Completeness} 
The following result shows that verification with the \qdtl calculus always produces correct results about the safety of distributed hybrid systems, i.e., the \qdtl calculus is sound.
\begin{theorem}[Soundness of \qdtl]\label{theo:soundness}
The {\em \qdtl} calculus is sound, i.e., every {\em \qdtl} (state) formula that can be proven is valid.
\end{theorem}

The verification for temporal safety ($[\alpha]\Box \phi$ or $\langle \alpha \rangle \Diamond \phi$), temporal liveness ($[\alpha]\Diamond\phi$ or $\langle \alpha \rangle \Box \phi$), and non-temporal ($[\alpha]\phi$ or $\langle \alpha \rangle \phi$) fragments of distributed hybrid systems has {\em three independent sources} of undecidability. Thus, no verification technique can be effective. Hence, \qdtl cannot be effectively axiomatizable. Both its discrete and its continuous fragments alone are subject to G\"{o}del's incompleteness theorem~\cite{APlatzer08}. The fragment
with only structural and dimension-changing dynamics is not effective either, because it can encode two-counter machines. 

\qdl has been proved to be complete relative to quantified differential equations~\cite{DBLP:conf/csl/Platzer10,DBLP:conf/csl/Platzer10:TR}. Due to the modular construction of the \qdtl calculus, we can lift the relative completeness result from \qdl to \qdtl. We essentially show that \qdtl is complete relative to \qdl, which directly implies that \qdtl calculus is even complete relative to an oracle for the fragment of \qdtl that has only quantified differential equations in modalities. Again, we restrict our attention to homogeneous combinations of path and trace quantifiers like $[\alpha] \Box \phi$ or $\langle \alpha \rangle \Diamond \phi$.
\begin{theorem}[Relative Completeness of \qdtl]\label{theo:completeness}
The calculus in Fig.~\ref{fig:rules} is a complete
axiomatization of {\em \qdtl} relative to quantified differential equations.
\end{theorem}
This result shows that both temporal and non-temporal properties of distributed hybrid systems can be proven to exactly the same extent to which properties of quantified differential equations can be proven. It also gives a formal justification that the \qdtl calculus reduces temporal properties to non-temporal \qdl properties.
\section{Verification of Distributed Air Traffic Control Safety Properties}\label{sec:7} 
Continuing the distributed air traffic control study from Section~\ref{sec:5}, the \qdtl proofs of the temporal safety invariant in (\ref{eq:safety}) and the temporal safety property in (\ref{eq:safetyproperty}) are presented in Fig.~\ref{fig:proofofverificationsafety} and Fig.~\ref{fig:proofofverificationsafety2}, respectively (for the purpose of simplifying the presentation, we ignore typing information $A$ for aircraft in the proof, because it is clear from the context). Note that temporal and non-temporal properties of the maneuver cannot be proven using any hybrid systems verification technique, because the dimension is parametric and unbounded and may even change dynamically during the remainder of the maneuver. The single proof in Fig.~\ref{fig:proofofverificationsafety} or Fig.~\ref{fig:proofofverificationsafety2} corresponds to infinitely many proofs for systems with $n$ aircraft for all $n$.

Our proofs show that the distributed roundabout maneuver always safely avoids collisions for arbitrarily many aircraft (even with dynamic appearance of new aircraft). The above maneuver still requires all aircraft in the horizon of relevance to participate in the collision avoidance maneuver. In fact, we can show that this is unnecessary for aircraft that are far enough away and that may be engaged in other roundabouts. Yet, this is beyond the scope of this paper.
\begin{figure}[!ht]
\vspace{-3mm}
\hspace*{-0.1in}
\begin{minipage}{1.1\linewidth}
{\scriptsize
\begin{prooftree}
\AxiomC{{\em true}}
\LeftLabel{$\mathbb{R}$}
\UnaryInfC{$\forall i, j \kern 0.15em \mathcal{T}(i, j) \rightarrow \forall i, j (2(x_1(i) - x_1(j))(-\omega(x_2(i) - x_2(j))) + 2(x_2(i) - x_2(j))\omega(x_1(i) - x_1(j)) \geq 0$)}
\LeftLabel{$\mathbb{R}$}
\UnaryInfC{$\forall i, j \kern 0.15em \mathcal{T}(i,j) \rightarrow \forall i, j ( 0 = 0 \wedge 2(x_1(i) - x_1(j))(d_1(i) - d_1(j)) + 2(x_2(i) - x_2(j))(d_2(i) - d_2(j)) \geq 0)$}
\LeftLabel{$[:=]$}
\UnaryInfC{$\forall i, j \kern 0.15em \mathcal{T}(i,j) \rightarrow [\forall i \kern 0.15em \mathcal{L}(i)] \forall i, j (i' = j' \wedge 2(x_1(i) - x_1(j))(x_1(i)' - x_1(j)') + 2(x_2(i) - x_2(j))(x_2(i)' - x_2(j)') \geq 0)$}
\end{prooftree}}
\end{minipage}
\end{figure}
\begin{figure}
\vspace{-14mm}
\hspace*{-0.6in}
\begin{minipage}{0.4\linewidth}
{\scriptsize
\begin{prooftree}
\AxiomC{{\em true}}
\LeftLabel{$\mathbb{R}$}
\UnaryInfC{$\forall i, j ( - \omega d_2(i) - (- \omega d_2(j)) = - \omega(d_2(i) - d_2(j))  \wedge \omega d_1(i) - \omega d_1(j) = \omega(d_1(i) - d_1(j)))$}
\LeftLabel{$[:=]$}
\UnaryInfC{$[\forall i \kern 0.15em \mathcal{L}(i)] \forall i, j ( d_1(i)' - d_1(j)' = - \omega (x_2(i)' - x_2(j)') \wedge d_2(i)' - d_2(j)' = \omega(x_1(i)' - x_1(j)' ))$}
\end{prooftree}}
\end{minipage}
\end{figure}
\begin{figure}
\vspace{-8mm}
\setlength{\unitlength}{0.75mm}
\begin{picture}(0,0)
\thicklines
\put(28,12){\vector(-1,4){2.6}}
\put(137,12){\vector(1,2){12.5}}
\end{picture}
\hspace*{-0.7in}
\begin{minipage}{0.4\linewidth}
{\scriptsize
\begin{prooftree}
\AxiomC{$\cdots$}
\LeftLabel{$[:=]$}
\UnaryInfC{$[\forall i \kern 0.15em \mathcal{L}(i)](\forall i, j \kern 0.15em \mathcal{T}(i, j))'$}
\LeftLabel{$DI$}
\UnaryInfC{$\forall i, j \kern 0.15em \mathcal{P}(i, j) \wedge \forall i, j \kern 0.15em \mathcal{T}(i, j) \rightarrow [\forall i \kern 0.15em {\mathcal{F}}_{\omega}(i)]\forall i,j \kern 0.15em \mathcal{T}(i,j)$}
				
				\AxiomC{$\cdots$}
                \LeftLabel{$[:=]$}
				\UnaryInfC{$\forall i, j \kern 0.15em \mathcal{T}(i, j) \rightarrow [\forall i \kern 0.15em \mathcal{L}(i)](\forall i, j \kern 0.15em \mathcal{P}(i, j))'$}
				\LeftLabel{$DI$}
				\UnaryInfC{$\forall i, j \kern 0.15em \mathcal{P}(i, j) \wedge \forall i, j \kern 0.15em \mathcal{T}(i, j) \rightarrow [\forall i \kern 0.15em {\mathcal{F}}_{\omega}(i) \with \forall i, j \kern 0.15em \mathcal{T}(i, j)]\forall i,j \kern 0.15em \mathcal{P}(i,j)$}
				
\LeftLabel{$DC$}
\BinaryInfC{$\forall i, j \kern 0.15em \mathcal{P}(i, j) \wedge \forall i, j \kern 0.15em \mathcal{T}(i, j) \rightarrow [\forall i \kern 0.15em {\mathcal{F}}_{\omega}(i)]\forall i,j \kern 0.15em \mathcal{P}(i,j)$}
\LeftLabel{$[']${\tiny $\Box$}}
\UnaryInfC{$\forall i, j \kern 0.15em \mathcal{P}(i, j) \wedge \forall i, j \kern 0.15em \mathcal{T}(i, j) \rightarrow [\forall i \kern 0.15em {\mathcal{F}}_{\omega}(i)] \Box \kern 0.08em \forall i,j \kern 0.15em \mathcal{P}(i,j)$}
\end{prooftree}}
\end{minipage}
\begin{center}
{\scriptsize Abbreviation: $\mathcal{L}(i) \equiv x_1(i)' : = d_1(i), x_2(i)' : = d_2(i), d_1(i)' := -\omega d_2(i), d_2(i)' := \omega d_1(i)$}
\end{center}
\vspace{-2mm}
\caption{\scriptsize{Proof for temporal collision freedom of roundabout collision avoidance maneuver circle}}
\label{fig:proofofverificationsafety}
\end{figure}

\begin{figure}[t!]
\vspace{-3mm}
\hspace*{-0.1in}
\begin{minipage}{1.1\linewidth}
{\scriptsize
\begin{prooftree}
\AxiomC{{\em true}}
\LeftLabel{$\mathbb{R}$}
\UnaryInfC{$\chi \wedge \forall i, j \kern 0.15em \mathcal{T}(i, j) \rightarrow \forall i, j (2(x_1(i) - x_1(j))(-\omega(x_2(i) - x_2(j))) + 2(x_2(i) - x_2(j))\omega(x_1(i) - x_1(j)) \geq 0$)}
\LeftLabel{$\mathbb{R}$}
\UnaryInfC{$\chi \wedge \forall i, j \kern 0.15em \mathcal{T}(i,j) \rightarrow \forall i, j ( 0 = 0 \wedge 2(x_1(i) - x_1(j))(d_1(i) - d_1(j)) + 2(x_2(i) - x_2(j))(d_2(i) - d_2(j)) \geq 0)$}
\LeftLabel{$[:=]$}
\UnaryInfC{$\chi \wedge \forall i, j \kern 0.15em \mathcal{T}(i,j) \rightarrow [\forall i \kern 0.15em \mathcal{K}(i)] \forall i, j (i' = j' \wedge 2(x_1(i) - x_1(j))(x_1(i)' - x_1(j)') + 2(x_2(i) - x_2(j))(x_2(i)' - x_2(j)') \geq 0)$}
\end{prooftree}}
\end{minipage}
\end{figure}
\begin{figure}
\hspace*{-1.0in}
\begin{minipage}{0.4\linewidth}
{\scriptsize
\begin{prooftree}
\AxiomC{{\em true}}
\LeftLabel{$\mathbb{R}$}
\UnaryInfC{$\chi \rightarrow \forall i, j ( - \omega d_2(i) - (- \omega d_2(j)) = - \omega(d_2(i) - d_2(j))  \wedge \omega d_1(i) - \omega d_1(j) = \omega(d_1(i) - d_1(j)))$}
\LeftLabel{$[:=]$}
\UnaryInfC{$\chi \rightarrow [\forall i \kern 0.15em \mathcal{K}(i)] \forall i, j ( d_1(i)' - d_1(j)' = - \omega (x_2(i)' - x_2(j)') \wedge d_2(i)' - d_2(j)' = \omega(x_1(i)' - x_1(j)' ))$}
\end{prooftree}}
\end{minipage}
\end{figure}
\begin{figure}[!ht]
\setlength{\unitlength}{0.75mm}
\begin{picture}(0,0)
\thicklines
\put(24.7,9.8){\vector(-1,4){1.6}}
\put(139.6,9.8){\vector(1,2){12.7}}
\end{picture}
\hspace*{-0.9in}
\begin{minipage}{0.4\linewidth}
{\scriptsize
\begin{prooftree}
\AxiomC{$\cdots$}
\LeftLabel{$[:=]$}
\UnaryInfC{$\chi
\rightarrow [\forall i \kern 0.15em
\mathcal{K}(i)] (\forall i, j \kern 0.15em \mathcal{T}(i,j))'$}
\LeftLabel{$DI$}
\UnaryInfC{$\forall i, j \kern 0.15em \mathcal{P}(i,j) \wedge \forall i, j \kern 0.15em
\mathcal{T}(i,j)
\rightarrow [\forall i \kern 0.15em
\mathcal{M}(i) \with \chi] \forall i, j \kern 0.15em \mathcal{T}(i,j)$}

                                                                                         \AxiomC{$\cdots$}
                                                                                         \LeftLabel{$[:=]$}
                                                                                         \UnaryInfC{$\chi \wedge \forall i, j \kern 0.15em \mathcal{T}(i,j) \rightarrow [\forall i \kern 0.15em \mathcal{K}(i)] (\forall i, j \kern 0.15em \mathcal{P}(i,j))'$}
                                                                                         \LeftLabel{$DI$}
                                                                                         \UnaryInfC{$\forall i, j \kern 0.15em \mathcal{P}(i,j) \wedge \forall i, j \kern 0.15em \mathcal{T}(i,j) \rightarrow [\forall i \kern 0.15em \mathcal{M}(i) \with \chi \wedge \forall i, j \kern 0.15em \mathcal{T}(i,j)] \forall i, j \kern 0.15em \mathcal{P}(i,j)$}
\LeftLabel{$DC$}
\BinaryInfC{$\forall i, j \kern 0.15em \mathcal{P}(i,j) \wedge \forall i, j \kern 0.15em
\mathcal{T}(i,j)
\rightarrow [\forall i \kern 0.15em
\mathcal{M}(i) \with \chi] \forall i, j \kern 0.15em \mathcal{P}(i,j)$}
\end{prooftree} }
\end{minipage}
\end{figure}
\begin{figure}[!ht]
\setlength{\unitlength}{0.75mm}
\begin{picture}(0,0)
\thicklines
\put(14.0,9.4){\vector(1,2){4.9}}
\put(147.4,9.4){\vector(-1,2){4.9}}
\end{picture}
\hspace*{-1.7in}
\begin{minipage}{0.4\linewidth}
{\scriptsize
\begin{prooftree}
\AxiomC{$\cdots$}
\UnaryInfC{$\forall i, j \kern 0.15em \mathcal{P}(i,j) \wedge \forall i, j \kern 0.15em
\mathcal{T}(i,j)
\rightarrow [\forall i \kern 0.15em
\mathcal{M}(i) \with \chi] \forall i, j \kern 0.15em \mathcal{P}(i,j)$}
\LeftLabel{$[']${\tiny $\Box$},$[:=]$}
\UnaryInfC{$\forall i, j \kern 0.15em \mathcal{P}(i,j) \wedge \forall i, j \kern 0.15em
\mathcal{T}(i,j)
\rightarrow [\forall i \kern 0.15em t := 0][\forall i \kern 0.15em
\mathcal{M}(i) \with \chi]
\Box \kern 0.08em \forall i, j \kern 0.15em \mathcal{P}(i,j)$}
                                                                                        \AxiomC{$\cdots$}
                                                                                        \UnaryInfC{$\forall i, j \kern 0.15em \mathcal{P}(i,j) \wedge \forall i, j \kern 0.15em \mathcal{T}(i,j) \rightarrow [\forall i \kern 0.15em \mathcal{M}(i) \with \chi] \forall i, j \kern 0.15em \mathcal{P}(i,j)$}
                                                                                        \LeftLabel{$[?]${\tiny $\Box$},$[:=]$}
                                                                                        \UnaryInfC{$\forall i, j \kern 0.15em \mathcal{P}(i,j) \wedge \forall i, j \kern 0.15em \mathcal{T}(i,j) \rightarrow [\forall i \kern 0.15em t := 0][\forall i \kern 0.15em \mathcal{M}(i) \with \chi][? \eta] \Box \kern 0.08em \forall i, j \kern 0.15em \mathcal{P}(i,j)$}
\LeftLabel{$[;]${\tiny $\Box$}}
\BinaryInfC{$\forall i, j \kern 0.15em \mathcal{P}(i,j) \wedge \forall i, j \kern 0.15em \mathcal{T}(i,j) \rightarrow [\forall i \kern 0.15em t := 0][\forall i \kern 0.15em \mathcal{M}(i) \with \chi; ? \eta] \Box \kern 0.08em \forall i, j \kern 0.15em \mathcal{P}(i,j)$}
\end{prooftree}}
\end{minipage}
\end{figure}
\begin{figure}[!ht]
\hspace*{-1.7in}
\begin{minipage}{0.4\linewidth}
{\scriptsize
\begin{prooftree}
\AxiomC{{\em true}}
\LeftLabel{{\em ax}}
\UnaryInfC{$\forall i, j \kern 0.15em \mathcal{P}(i,j), \forall i, j \kern 0.15em \mathcal{T}(i,j) \rightarrow \forall i, j \kern 0.15em \mathcal{P}(i,j)$}
\LeftLabel{$\wedge$l}
\UnaryInfC{$\forall i, j \kern 0.15em \mathcal{P}(i,j) \wedge \forall i, j \kern 0.15em \mathcal{T}(i,j) \rightarrow \forall i, j \kern 0.15em \mathcal{P}(i,j)$}
                                                              \AxiomC{{\em true}}
                                                              \LeftLabel{{\em ax}}
                                                              \UnaryInfC{$\forall i, j \kern 0.15em \mathcal{P}(i,j), \forall i, j \kern 0.15em \mathcal{T}(i,j) \rightarrow \forall i, j \kern 0.15em \mathcal{P}(i,j)$}
                                                              \LeftLabel{$[:=]$,$\wedge$l}
                                                              \UnaryInfC{$\forall i, j \kern 0.15em \mathcal{P}(i,j) \wedge \forall i, j \kern 0.15em \mathcal{T}(i,j) \rightarrow [\forall i \kern 0.15em t := 0] \forall i, j \kern 0.15em \mathcal{P}(i,j)$}
\LeftLabel{$[:=]${\tiny $\Box$}}
\BinaryInfC{$\forall i, j \kern 0.15em \mathcal{P}(i,j) \wedge \forall i, j \kern 0.15em \mathcal{T}(i,j) \rightarrow [\forall i \kern 0.15em t := 0]
\Box \kern 0.08em \forall i, j \kern 0.15em \mathcal{P}(i,j)$}
\end{prooftree}}
\end{minipage}
\end{figure}
\begin{figure}[!ht]
\setlength{\unitlength}{0.75mm}
\begin{picture}(0,0)
\thicklines
\put(23.5,7){\vector(1,2){2.9}}
\put(131.4,7){\vector(1,2){16.2}}
\end{picture}
\hspace*{-0.9in}
\begin{minipage}{0.4\linewidth}
{\scriptsize
\begin{prooftree}
\AxiomC{$\cdots$}
\UnaryInfC{$\forall i, j \kern 0.15em \mathcal{P}(i,j) \wedge \forall i, j
\kern 0.15em \mathcal{T}(i,j) \rightarrow
[\forall i \kern 0.15em t := 0] \Box \kern 0.08em
\forall i, j \kern 0.15em \mathcal{P}(i,j)$}
                                                  \AxiomC{$\cdots$}
                                                  \UnaryInfC{$\forall i, j \kern 0.15em \mathcal{P}(i,j) \wedge \forall i, j \kern 0.15em \mathcal{T}(i,j) \rightarrow [\forall i \kern 0.15em t := 0][\forall i \kern 0.15em \mathcal{M}(i) \with \chi; ?\eta] \Box \kern 0.08em \forall i, j \kern 0.15em \mathcal{P}(i,j)$}
\LeftLabel{$[;]${\tiny $\Box$}}
\BinaryInfC{$\forall i, j \kern 0.15em \mathcal{P}(i,j) \wedge \forall i, j \kern 0.15em \mathcal{T}(i,j) \rightarrow [\forall i \kern 0.15em t := 0; \forall i \kern 0.15em \mathcal{M}(i) \with \chi; ?\eta] \Box \kern 0.08em \forall i, j \kern 0.15em \mathcal{P}(i,j)$}
\end{prooftree}}
\end{minipage}
\end{figure}
\begin{figure}[!ht]
\[{\scriptsize
\begin{array}{l c}
\text{Abbreviations:} & \mathcal{M}(i) \equiv {\mathcal{F}}_{\omega(i)}(i), t' = 1 \\
& \chi \equiv \forall i \kern 0.15em t \leq T \\
& \eta \equiv \forall i \kern 0.15em t = T \\
& \mathcal{K}(i) \equiv x_1(i)' : = d_1(i), x_2(i)' : = d_2(i), d_1(i)' := -\omega d_2(i), d_2(i)' := \omega d_1(i), t' := 1 \\
\end{array}}\]
\caption{\scriptsize{Proof for temporal collision freedom of roundabout collision avoidance maneuver circle in bounded time}}
\label{fig:proofofverificationsafety2}
\end{figure}

\section{Liveness by Quantifier Alternation}\label{sec:8} 
Liveness specifications of the form $[\alpha]\Diamond \phi$ or $\langle \alpha \rangle \Box \phi$ are sophisticated ($\Sigma^1_1$-hard
because they can express infinite occurrence in Turing machines). Beckert and
Schlager~\cite{BeckertS01}, for instance, note that they failed to find sound rules for a discrete case that corresponds
to $[\alpha;\beta]\Diamond \phi$.

For {\em finitary liveness semantics}, we can still find proof rules. In this section,
we modify the meaning of $[\alpha]\Diamond\phi$ to refer to all {\em terminating} traces of $\alpha$. Then,
the straightforward generalization $[;]\diamond$ in Fig.~\ref{fig:rules2} is sound, even in the hybrid
case. But $[;]\diamond$ still leads to an incomplete axiomatization
as it does not cover the case where, in some traces, $\phi$  becomes true at some
point during $\alpha$, and in other traces, $\phi$ only becomes true during
$\beta$. To overcome this limitation, we use a program transformation approach. We instrument the
quantified hybrid program to monitor the occurrence of $\phi$ during all changes: In $[\alpha]\diamond$, $\check{\alpha}$ results from replacing all occurrences of $\forall i : \kern -0.17em A \ f(\boldsymbol{s}) := \theta$ with $\forall i : \kern -0.17em A \ f(\boldsymbol{s}) := \theta; ? (\phi \rightarrow t = 1)$ and $\forall i : \kern -0.17em A \ f(\boldsymbol{s})' = \theta \with \chi$
with $\forall i : \kern -0.17em A \ f(\boldsymbol{s})' = \theta \with \chi \wedge (\phi \rightarrow t =1)$. The latter is a continuous evolution restricted to the region that satisfies $\chi$ and $\phi \rightarrow t = 1$. The effect
is that $t$ detects whether $\phi$ has occurred during any change in $\alpha$. In particular, $t$
is guaranteed to be $1$ after all runs if $\phi$ occurs at least once along all traces of $\alpha$.
\begin{figure}
\vspace{-5mm}
{\scriptsize \[([;]\diamond) \ \inferencerule{\rightarrow [\alpha]\Diamond \phi, [\alpha][\beta]\Diamond\phi}{\rightarrow [\alpha; \beta]\Diamond \phi} \quad \quad
([\alpha]\diamond) \ \inferencerule{\phi \vee \forall t : \kern -0.17em \mathbb{R}\ [\check{\alpha}]t = 1}{[\alpha]\Diamond\phi} \]}
\vspace{-3mm}
\caption{{\scriptsize Transformation rules for alternating temporal path and trace quantifiers}}
\label{fig:rules2}
\end{figure}
\section{Conclusions and Future Work}\label{sec:9} 
For reasoning about distributed hybrid systems, we have introduced a temporal dynamic
logic, \qdtl, with modal path quantifiers over traces and temporal quantifiers
along the traces. It combines the capabilities of quantified differential dynamic logic to reason
about possible distributed hybrid system behavior with the power of temporal logic in reasoning about the behavior along traces. Furthermore, we have presented a proof calculus for verifying temporal safety specifications of quantified hybrid programs in \qdtl.

Our sequent calculus for \qdtl is a completely modular combination of temporal and non-temporal
reasoning. Temporal formulas are handled using rules that augment
intermediate state transitions with corresponding sub-specifications. Purely non-temporal
rules handle the effects of discrete transitions, continuous evolutions, and structural/dimension changes. The modular nature of the \qdtl calculus further enables us to lift the relative completeness result from \qdl to \qdtl. This theoretical result shows that the \qdtl calculus is a sound and complete axiomatization of the temporal behavior of distributed hybrid systems relative to differential equations. As an example, we demonstrate that our logic is suitable for reasoning about temporal safety properties in a distributed air traffic control system.

We are currently extending our verification tool for distributed
hybrid systems, which is an automated theorem prover called KeYmaeraD~\cite{DBLP:conf/icfem/RenshawLP11}, to cover the full \qdtl calculus. Future work includes extending
\qdtl with $\text{CTL}^{\ast}$-like~\cite{EmersonH86} formulas of the form $[\alpha](\psi \wedge \Box \phi)$ to avoid splitting of
the proof into two very similar sub-proofs for temporal parts $[\alpha]\Box\phi$ and non-temporal
parts $[\alpha]\psi$  arising in $[;]${\tiny $\Box$}. Our combination of temporal logic with dynamic
logic is more suitable for this purpose than the approach in~\cite{BeckertS01}, since \qdtl
has uniform modalities and uniform semantics for temporal and non-temporal
specifications. This extension will also simplify the treatment of alternating liveness
quantifiers conceptually.
\bibliographystyle{plain}
\bibliography{cmulib}

\begin{thebibliography}{10}

\bibitem{AlurCD90}
R.~Alur, C.~Courcoubetis, and D.~L. Dill.
\newblock Model-checking for real-time systems.
\newblock In {\em LICS}, pages 414--425. IEEE Computer Society, 1990.

\bibitem{BeckertS01}
B.~Beckert and S.~Schlager.
\newblock A sequent calculus for first-order dynamic logic with trace
  modalities.
\newblock In R.~Gor\'{e}, A.~Leitsch, and T.~Nipkow, editors, {\em IJCAR},
  volume 4130 of {\em LNCS}, pages 626--641. Springer, 2001.

\bibitem{ClarkeGP99}
E.~M. Clarke, O.~Grumberg, and D.~A. Peled.
\newblock {\em Model Checking}.
\newblock MIT Press, Cambridge, MA, USA, 1999.

\bibitem{DammHO03}
W.~Damm, H.~Hungar, and E.~R. Olderog.
\newblock On the verification of cooperating traffic agents.
\newblock In F.~S. de~Boer, M.~M. Bonsangue, S.~Graf, and W.~P. de~Roever,
  editors, {\em FMCO}, volume 3188 of {\em LNCS}, pages 77--110. Springer,
  2003.

\bibitem{DavorenCM04}
J.~M. Davoren, V.~Coulthard, N.~Markey, and T.~Moor.
\newblock Non-deterministic temporal logics for general flow systems.
\newblock In R.~Alur and G.~J. Pappas, editors, {\em HSCC}, volume 2993 of {\em
  LNCS}, pages 280--295. Springer, 2004.

\bibitem{DavorenN00}
J.~M. Davoren and A.~Nerode.
\newblock Logics for hybrid systems.
\newblock {\em Proceedings of the IEEE}, 88(7):985--1010, July 2000.

\bibitem{DeshpandeGV96}
A.~Deshpande, A.~G\"oll\"u, and P.~Varaiya.
\newblock {SHIFT}: A formalism and a programming language for dynamic networks
  of hybrid automata.
\newblock In {\em Hybrid Systems}, pages 113--133, 1996.

\bibitem{DowekMC05}
G.~Dowek, C.~Mu{\~n}oz, and V.~A. Carre{\~n}o.
\newblock Provably safe coordinated strategy for distributed conflict
  resolution.
\newblock In {\em AIAA Proceedings, AIAA-2005-6047}, pages 278--292, 2005.

\bibitem{EmersonC82}
E.~A. Emerson and E.~M. Clarke.
\newblock Using branching time temporal logic to synthesize synchronization
  skeletons.
\newblock {\em Sci. Comput. Program}, 2(3):241--266, 1982.

\bibitem{EmersonH86}
E.~A. Emerson and J.~Y. Halpern.
\newblock ``{S}ometimes" and ``{N}ot {N}ever" revisited: on branching versus
  linear time temporal logic.
\newblock {\em J. ACM}, 33(1):151--178, 1986.

\bibitem{FaberM06}
J.~Faber and R.~Meyer.
\newblock Model checking data-dependent real-time properties of the {E}uropean
  {T}rain {C}ontrol {S}ystem.
\newblock In {\em FMCAD}, pages 76--77. IEEE Computer Society Press, Nov 2006.

\bibitem{HarelKT00}
D.~Harel, D.~Kozen, and J.~Tiuryn.
\newblock {\em Dynamic logic}.
\newblock MIT Press, Cambridge, 2000.

\bibitem{Henzinger96}
T.~A. Henzinger.
\newblock The theory of hybrid automata.
\newblock In {\em LICS}, pages 278--292, 1996.

\bibitem{HenzingerNSY92}
T.~A. Henzinger, X.~Nicollin, J.~Sifakis, and S.~Yovine.
\newblock Symbolic model checking for real-time systems.
\newblock In {\em LICS}, pages 394--406. IEEE Computer Society, 1992.

\bibitem{HsuESV91}
A.~Hsu, F.~Eskafi, S.~Sachs, and P.~Varaiya.
\newblock Design of platoon maneuver protocols for {IVHS}.
\newblock Technical Report PATH Research Report UCB-ITS-PRR-91-6, UC Berkeley,
  1991.

\bibitem{KratzSPL06}
F.~Kratz, O.~Sokolsky, G.~J. Pappas, and I.~Lee.
\newblock R-{C}haron, a modeling language for reconfigurable hybrid systems.
\newblock In {\em HSCC}, pages 392--406, 2006.

\bibitem{Leivan04}
D.~Leivant.
\newblock Partial correctness assertions provable in dynamic logics.
\newblock In I.~Walukiewicz, editor, {\em FoSSaCS}, volume 2987 of {\em LNCS},
  pages 304--317. Springer, 2004.

\bibitem{MysorePM05}
V.~Mysore, C.~Piazza, and B.~Mishra.
\newblock Algorithmic algebraic model checking {II}: {D}ecidability of
  semi-algebraic model checking and its applications to systems biology.
\newblock In D.~Peled and Y.~K. Tsay, editors, {\em ATVA}, volume 3707 of {\em
  LNCS}, pages 217--233. Springer, 2005.

\bibitem{APlatzer08}
A.~Platzer.
\newblock Differential dynamic logic for hybrid systems.
\newblock {\em J. Autom. Reas.}, 41(2):143--189, 2008.

\bibitem{APlatzer10}
A.~Platzer.
\newblock Differential-algebraic dynamic logic for differential-algebraic
  programs.
\newblock {\em J. Log. Comput.}, 20(1):309--352, 2010.

\bibitem{DBLP:conf/csl/Platzer10}
A.~Platzer.
\newblock Quantified differential dynamic logic for distributed hybrid systems.
\newblock In Anuj Dawar and Helmut Veith, editors, {\em CSL}, volume 6247 of
  {\em LNCS}, pages 469--483. Springer, 2010.

\bibitem{DBLP:conf/csl/Platzer10:TR}
A.~Platzer.
\newblock Quantified differential dynamic logic for distributed hybrid systems.
\newblock Technical Report CMU-CS-10-126, School of Computer Science, Carnegie
  Mellon University, Pittsburgh, PA, May 2010.

\bibitem{Pnueli77}
A.~Pnueli.
\newblock The temporal logic of programs.
\newblock In {\em FOCS}, pages 46--57. IEEE, 1977.

\bibitem{Pratt79}
V.~R. Pratt.
\newblock Process logic.
\newblock In {\em POPL}, pages 93--100, 1979.

\bibitem{DBLP:conf/icfem/RenshawLP11}
D.~W. Renshaw, S.~M. Loos, and A.~Platzer.
\newblock Distributed theorem proving for distributed hybrid systems.
\newblock In Shengchao Qin and Zongyan Qiu, editors, {\em ICFEM}, volume 6991
  of {\em LNCS}, pages 356--371. Springer, 2011.

\bibitem{TomlinPS98}
C.~Tomlin, G.~J. Pappas, and S.~Sastry.
\newblock Conflict resolution for air traffic management: a study in
  multi-agent hybrid systems.
\newblock {\em IEEE T. Automat. Contr.}, 43(4):509--521, 1998.

\bibitem{BeekMRRS06}
D.~A. van Beek, K.~L. Man, M.~A. Reniers, J.~E. Rooda, and R.~R.~H.
  Schiffelers.
\newblock Syntax and consistent equation semantics of hybrid chi.
\newblock {\em J. Log. Algebr. Program}, 68(1-2):129--210, 2006.

\end{thebibliography}
\newpage
\appendix
\section{Proof of Conservative Temporal Extension}
\subsection{Reachability Semantics of \qdl}
In this subsection, we review the reachability semantics of \qdl given in~\cite{DBLP:conf/csl/Platzer10,DBLP:conf/csl/Platzer10:TR} before proving Proposition~\ref{prop:cextension} in the next subsection.
\begin{definition}[Reachability Semantics of QHP] \label{def:semanticsofqhp}
The reachability semantics, $\rho(\alpha)$, of QHP $\alpha$, is {\em a transition relation} on states. It
specifies which state $\tau \in \mathcal{S}$ is reachable from a state $\sigma \in \mathcal{S}$ by running QHP $\alpha$ and is defined as:
\begin{enumerate}
\item $(\sigma, \tau) \in \rho(\forall i : \kern -0.17em A \ f(\boldsymbol{s}) := \theta)$ iff $\tau$ is identical to $\sigma$ except that at each position $\boldsymbol{o}$ of $f$: if $\sigma_i^{e} \llbracket \boldsymbol{s} \rrbracket = \boldsymbol{o}$ for some object $e \in \sigma(A)$, then $\tau(f)(\sigma_i^e\llbracket \boldsymbol{s} \rrbracket) = \sigma_i^e \llbracket \theta \rrbracket$. If there are multiple objects $e$ giving the same position $\sigma_i^{e} \llbracket \boldsymbol{s} \rrbracket = \boldsymbol{o}$,  then all of the resulting states $\tau$ are reachable.

\item $(\sigma, \tau) \in \rho(\forall i : \kern -0.17em A \ f(\boldsymbol{s})' = \theta \with \chi)$ iff there is a function $\varphi: [0, r] \rightarrow \mathcal{S}$ for some $r \geq 0$ with $\varphi(0) = \sigma$ and $\varphi(r) =\tau$ satisfying the following conditions. At each time $t \in [0,r]$, state $\varphi(t)$ is identical to $\sigma$, except that at each position $\boldsymbol{o} \ \text{of} \ f: \ \text{if} \ \sigma_i^{e} \llbracket \boldsymbol{s} \rrbracket = \boldsymbol{o}$ for some object $e \in \sigma(A)$, then, at each time $\zeta \in [0,r]$:
     \begin{itemize}
     \item The differential equations hold and derivatives exist (trivial for $r = 0$):
     \[\inferencerule{\mathsf{d} ({\varphi(t)}^{e}_{i}\llbracket f(\boldsymbol{s}) \rrbracket)}{\mathsf{d} t}(\zeta) = ({\varphi(\zeta)}^{e}_{i}\llbracket \theta \rrbracket) \]
     \item The evolution domains is respected: ${\varphi(\zeta)}^{e}_{i} \llbracket \chi \rrbracket = \text{true}.$
     \end{itemize}
     If there are multiple objects $e$ giving the same position $\sigma_i^{e} \llbracket \boldsymbol{s} \rrbracket = \boldsymbol{o}$,  then all of the resulting states $\tau$ are reachable.

 \item $\rho(?\chi) = \{ (\sigma, \sigma) : \sigma \llbracket \chi \rrbracket = \text{true} \}$
 \item $\rho(\alpha \cup \beta) = \rho(\alpha) \cup \rho(\beta)$
 \item $\rho(\alpha;\beta) = \{ (\sigma, \tau) : (\sigma, z) \in \rho(\alpha)$ and $(z, \tau) \in \rho(\beta)$ for a state $z. \}$

 \item $(\sigma, \tau) \in \rho(\alpha^{\ast})$ iff there is an $n \in  \mathbb{N}$ with $n \geq 0$ and there are states $\sigma = \sigma_0, \ldots, \sigma_n = \tau$ such that $(\sigma_i, \sigma_{i+1}) \in  \rho(\alpha)$ for all $0 \leq i < n$.
\end{enumerate}
\end{definition}

\begin{definition}[Valuation of \qdl Formulas]\label{def:semanticsofqdl}
 The valuation of \qdl formula $\phi$ with respect to state $\sigma$ is defined as follows:
\begin{enumerate}
\item $\sigma \llbracket \theta_1 = \theta_2 \rrbracket = $ true  iff $\sigma \llbracket \theta_1 \rrbracket = \sigma \llbracket \theta_2 \rrbracket$; accordingly for $\geq$.
\item $\sigma \llbracket \phi \wedge \psi \rrbracket =$ true iff $\sigma \llbracket \phi \rrbracket =$ true and $\sigma \llbracket \psi \rrbracket =$ true; accordingly for $\neg$.
\item $\sigma \llbracket \forall i : \kern -0.17em A \ \phi \rrbracket =$ true iff $\sigma_i^e \llbracket \phi \rrbracket =$ true for all objects $e \in \sigma(A)$.
\item $\sigma \llbracket \exists i : \kern -0.17em A \ \phi \rrbracket =$ true iff $\sigma_i^e \llbracket \phi \rrbracket =$ true for some object $e \in \sigma(A)$.
\item $\sigma \llbracket [\alpha] \phi \rrbracket =$ true iff $\tau \llbracket \phi \rrbracket =$ true for all states $\tau$ with $(\sigma, \tau) \in \rho(\alpha)$.

\item $\sigma \llbracket \langle \alpha \rangle \phi \rrbracket =$ true iff $\tau \llbracket \phi \rrbracket =$ true for some $\tau$ with $(\sigma, \tau) \in \rho(\alpha)$.
\end{enumerate}
\end{definition}

\subsection{Extension Proof}
The proof of Proposition~\ref{prop:cextension} uses the following lemma about the relationship of reachability and trace semantics of \qdtl programs, which agree on initial and final states.
\begin{lemma}
For QHP $\alpha$, we have
\[\rho(\alpha) = \{(\emph{\text{first}} \kern 0.25em \nu, \emph{\text{last}} \kern 0.25em \nu): \nu \in \mu(\alpha) \ \text{terminates}. \} \]
\end{lemma}

\begin{proof}
The proof follows an induction on the structure of $\alpha$.
\begin{itemize}
\item The cases $\forall i : \kern -0.17em A \ f(\boldsymbol{s}) := \theta$, $\forall i : \kern -0.17em A \ f(\boldsymbol{s})' = \theta \with \chi$, and $\alpha \cup \beta$ are simple by comparing the Definition 2 and Definition 4.
\item For $?\chi$, the reasoning splits into two directions. For the direction ``$\supseteq$", assume $\nu \in \mu(?\chi)$.
We distinguish between two cases. If first\kern 0.25em $\nu \llbracket \chi \rrbracket$ = {\em true}, then $\nu = (\hat{\sigma})$ has length one,
last\kern 0.25em$\nu$ = first\kern 0.25em$\nu$, and (first\kern 0.25em$\nu$, last\kern 0.25em$\nu$) $\in \rho(\alpha)$.
If, however, first\kern 0.25em$\nu \llbracket \chi \rrbracket$ = {\em false}, then $\nu = (\hat{\sigma}, \hat{\Lambda})$ does not terminate, hence, there is nothing to show. Conversely, for ``$\subseteq$", assume $(\sigma, \sigma) \in \rho(?\chi)$, then $\sigma \llbracket \chi \rrbracket$ = {\em true} and $(\hat{\sigma}) \in \mu(?\chi)$ satisfies the conditions on $\nu$.
\item For $\alpha ; \beta$  and the direction ``$\supseteq$", assume that $\nu \circ \varsigma \in \mu(\alpha; \beta)$ terminates with $\nu \in \mu(\alpha), \varsigma \in \mu(\beta)$, and last\kern 0.25em$\nu$ = first\kern 0.25em$\varsigma$. Then, by induction hypothesis, we can assume that (first\kern 0.25em$\nu$, last\kern 0.25em$\nu) \in \rho(\alpha)$, and (first\kern 0.25em$\varsigma$, last\kern 0.25em$\varsigma) \in \rho(\beta)$. By the semantics
of sequential composition, we conclude (first\kern 0.25em$\nu \circ \varsigma$, last\kern 0.25em$\nu \circ \varsigma) \in \rho(\alpha; \beta)$. Conversely, for ``$\subseteq$", assume that $(\sigma, \tau) \in \rho(\alpha; \beta)$, i.e., let $(\sigma, z) \in \rho(\alpha)$ and $(z, \tau) \in \rho(\beta)$. By induction hypothesis, there is a terminating trace $\nu \in \mu(\alpha)$ with first\kern 0.25em$\nu = \sigma$ and last\kern 0.25em$\nu = z$. Further, by induction hypothesis, there is a terminating trace $\varsigma \in \mu(\beta)$ with first\kern 0.25em$\varsigma = z$ and last\kern 0.25em$\varsigma = \tau$. Hence, $\nu \circ \varsigma \in \mu(\alpha; \beta)$ terminates, first\kern 0.25em$\nu \circ \varsigma = \sigma$ and last\kern 0.25em$\nu \circ \varsigma = \tau$.
\item The case $\alpha^{\ast}$ is an inductive consequence of the sequential composition case.
\end{itemize}
\end{proof}

\begin{proof}[of Proposition~\ref{prop:cextension}]
The formulas of \qdl are a subset of the \qdtl formulas. In the course of this proof, we use the notation $\sigma_{\qdl}\llbracket \cdot \rrbracket$ to indicate that the \qdl valuation from Definition~\ref{def:semanticsofqdl} is used. For \qdl formulas $\psi$, we show that the valuations with respect to Definition 3 and
with respect to Definition 5 are the same for all states $\sigma$:
\[\sigma \llbracket \psi \rrbracket = \sigma_{\qdl} \llbracket \psi \rrbracket \  \text{for all } \sigma. \]
We prove this by induction on the structure of $\psi$. The cases $\theta_1 = \theta_2, \theta_1 \geq \theta_2, \neg \phi, \phi \wedge \psi, \forall i : \kern -0.17em A \ \phi, \exists i : \kern -0.17em A \ \phi$ of state formulas are obvious. The other cases are proven as
follows.
\begin{itemize}
\item If $\psi$ has the form $[\alpha]\phi$, assume that $\sigma \llbracket [\alpha]\phi \rrbracket$ = {\em false}. Then there is some
terminating trace $\nu \in \mu(\alpha)$ with first\kern 0.25em$\nu = \sigma$ such that last\kern 0.25em$\nu \llbracket \phi \rrbracket$ = {\em false}. By induction hypothesis, this implies that last\kern 0.25em$\nu_{\qdl} \llbracket \phi \rrbracket$ = {\em false}. According
to Lemma 1, $(\sigma$, last\kern 0.25em$\nu) \in \rho(\alpha)$ holds, which implies $\sigma_{\qdl} \llbracket [\alpha]\phi \rrbracket$ = {\em false}. For the converse direction, assume that $\sigma_{\qdl} \llbracket [\alpha]\phi \rrbracket$ = {\em false}. Then there is a $(\sigma, \tau) \in \rho(\alpha)$ with $\tau_{\qdl} \llbracket \phi \rrbracket$ = {\em false}. By Lemma 1, there is a terminating trace $\nu \in \mu(\alpha)$ with first\kern 0.25em$\nu = \sigma$ and last\kern 0.25em$\nu = \tau$. By induction hypothesis, last\kern 0.25em$\nu \llbracket \phi \rrbracket = $ {\em false}. Thus, we can conclude that both $\nu \llbracket \phi \rrbracket = $ {\em false} and $\nu \llbracket [\alpha]\phi \rrbracket = $ {\em false}.
\item The case $\psi = \langle \alpha \rangle \phi$ is proven accordingly.
\end{itemize}
\end{proof}

\newpage

\section{Proof of Soundness and Completeness}

\subsection{Proof of Soundness}

\begin{proof}[Proof of Theorem~\ref{theo:soundness}]
We show that all rules of the \qdtl calculus are {\em locally sound}, i.e., for all states $\sigma$, the conclusion of a rule is true in state $\sigma$ when all
premisses are true in $\sigma$. Let $\sigma$ be any state. For each rule we
have to show that the conclusion is true in $\sigma$ assuming the premisses are true
in $\sigma$. Inductively, the soundness of the non-temporal rules follows from Proposition~\ref{prop:cextension} and local soundness of the corresponding rules in \qdl~\cite{DBLP:conf/csl/Platzer10,DBLP:conf/csl/Platzer10:TR}. The proof for the generalization in $[\cup]$ and $\langle \cup \rangle$ to path formulas $\pi$ is a straightforward extension.
\begin{itemize}
\item $[;]${\tiny $\Box$} \  Assume $\sigma \models [\alpha]\Box\phi$ and $\sigma \models [\alpha][\beta]\Box\phi$. Let $\nu \in \mu(\alpha;\beta)$, i.e., $\nu = \varrho \circ \varsigma$ with first\kern 0.25em$\nu$ = $\sigma$, $\varrho \in \mu(\alpha)$, and $\varsigma \in \mu(\beta)$. If $\varrho$ does not terminate, then $\nu = \varrho \in \mu(\alpha)$ and $\nu \models \Box \phi$ by premise. If $\varrho$ terminates with last\kern 0.25em$\varrho$ = first\kern 0.25em$\varsigma$, then $\varrho \models \Box \phi$ by premise. Further, we know $\sigma \models [\alpha][\beta]\Box\phi$. In particular for trace $\varrho \in \mu(\alpha)$, we have last\kern 0.25em$\varrho$ $\models [\beta]\Box\phi$. Thus, $\varsigma \models \Box \phi$ because $\varsigma \in \mu(\beta)$ starts at first\kern 0.25em$\varsigma$ = last\kern 0.25em$\varrho$. By composition, $\varrho \circ \varsigma \models \Box \phi$ . As $\nu = \varrho \circ \varsigma$ was arbitrary, we can conclude $\sigma \models [\alpha; \beta]\Box \phi$. The converse direction holds, as all traces of $\alpha$ are prefixes of traces of $\alpha; \beta$. Hence, the assumption $\sigma \models [\alpha;\beta]\Box \phi$ directly implies $\sigma \models [\alpha]\Box\phi$. Further, all traces of $\beta$ that begin at a state reachable from $\sigma$  by $\alpha$ are suffixes of traces of $\alpha;\beta$ that start in $\sigma$. Hence, $\sigma \models [\alpha][\beta]\Box\phi$ is implied as well.
\item $[?]${\tiny $\Box$} \ Soundness of $[?]${\tiny $\Box$} is obvious, since, by premise, we can assume $\sigma \models \phi$, and there is nothing to show for $\Lambda$ states according to Definition 3. Conversely, $\hat{\sigma}$ is a prefix of all traces in $\mu(?\chi)$ that start in $\sigma$.
\item $[:=]${\tiny $\Box$} \ Assuming $\sigma \models \phi$ and $\sigma \models [\forall i : \kern -0.17em A \ f(\boldsymbol{s}) := \theta] \phi$, we have to show that $\sigma \models [\forall i : \kern -0.17em A \ f(\boldsymbol{s}) := \theta] \Box \phi$. Let $\nu \in \mu(\forall i : \kern -0.17em A \ f(\boldsymbol{s}) := \theta)$ be any trace with first\kern 0.25em$\nu$ = $\sigma$, i.e., $\nu = (\hat{\sigma}, \hat{\tau})$ by Definition 2. Hence, the only two states we need to consider are $\nu_{0}(0) = \sigma$ and $\nu_{1}(0) = \tau$. By premise, $\nu_{0}(0) = \sigma$ yields $\nu_{0}(0) \models \phi$. Similarly, for the state $\nu_{1}(0) =$ last\kern 0.25em$\nu = \tau$, the premise gives $\nu_{1}(0) \models \phi$. The converse
    direction is similar.
\item $[']${\tiny $\Box$} \ We prove that $[']${\tiny $\Box$} is locally sound by contraposition. For this, assume that $\sigma \not \models [\forall i : \kern -0.17em A \ f(\boldsymbol{s})' = \theta \with \chi] \Box \phi$; then there is a trace $\nu = (\varphi) \in \mu(\forall i : \kern -0.17em A \ f(\boldsymbol{s})' = \theta \with \chi)$ starting in first\kern 0.25em$\nu = \sigma$ and $\sigma \not \models \Box \phi$. Hence, there is a position $(0, \zeta)$ of $\nu$ with $\nu_{0}(\zeta) \not \models \phi$. Now $\varphi$ restricted to $[0, \zeta]$ also solves the quantified differential equation $\forall i : \kern -0.17em A \ f(\boldsymbol{s})' = \theta \with \chi$. Thus, $(\varphi |_{[0, \zeta]}) \not \models \phi$ as $\varphi(\zeta) \not \models \phi$, since the last state is $\varphi(\zeta)$. By consequence, this gives $\sigma \not \models [\forall i : \kern -0.17em A \ f(\boldsymbol{s})' = \theta \with \chi] \phi$. The converse direction is obvious as last\kern 0.25em$\nu$ always is a state occurring during $\nu$. Hence,  $\sigma \not \models [\forall i : \kern -0.17em A \ f(\boldsymbol{s})' = \theta \with \chi] \phi$ immediately implies $\sigma \not \models [\forall i : \kern -0.17em A \ f(\boldsymbol{s})' = \theta \with \chi] \Box \phi$.
\item $[{}^{\ast n}]${\tiny $\Box$} \ By contraposition, assume that $\sigma \not \models [\alpha^{\ast}]\Box\phi$. Then there is an $n \in \mathbb{N}$ and a trace $\nu \in \mu(\alpha^{n})$ with first\kern 0.25em$\nu = \sigma$ such that $\nu \not \models \Box \phi$. There are two cases. If $n > 0$ then $\nu \in \mu(\alpha;\alpha^{\ast})$, and thus $\sigma \not \models [\alpha;\alpha^{\ast}]\Box \phi$. If, however, $n = 0$, then $\nu = (\hat{\sigma})$ and $\sigma \not \models \phi$. Hence, all traces $\varsigma \in \mu(\alpha; \alpha^{\ast})$ with first\kern 0,25em$\varsigma = \sigma$ satisfy $\varsigma \not \models \Box \phi$. Finally, it is easy to see that all programs have at least one such trace $\varsigma$ that witnesses $\sigma \not \models [\alpha; \alpha^{\ast}] \Box \phi$. The converse direction is easy as all behaviour of $\alpha; \alpha^{\ast}$ is subsumed by $\alpha^{\ast}$, i.e., $\mu(\alpha; \alpha^{\ast}) \subseteq \mu(\alpha^{\ast})$.
\item $[{}^{\ast}]${\tiny $\Box$} \ Clearly, using the fact that $\mu(\alpha^{\ast}) \supseteq \mu(\alpha^{\ast};\alpha)$, the set of states along the traces of $\alpha^{\ast}$ at which $\phi$ needs to be true for the premise is a subset of the corresponding set for the conclusion. Hence, the conclusion entails the premise. Conversely, all states during traces of $\alpha^{\ast}$ are also reachable by iterating $\alpha$ sufficiently often to completion and then following a single trace of $\alpha$. In detail: If $\sigma \not \models [\alpha^{\ast}] \Box \phi$. then there is a trace $\nu \in \mu(\alpha^{\ast})$ on which $\neg \phi$ holds true at some state, say, at $\nu_{k}(\zeta) \not = \Lambda$. Let $n \geq 0$ be the (maximum) number of complete repetitions of $\alpha$ along $\nu$ before discrete step index $k$. That is, there is some discrete step index $k_n < k$ such that the prefix $\varrho = (\nu_0, \ldots, \nu_{k_n}) \in \mu(\alpha^{n})$ of $\nu$ consists of $n$ complete repetitions of $\alpha$ and the suffix $\varsigma = (\nu_{k_{n+1}}, \nu_{k_{n+2}}, \ldots) \in \mu(\alpha^{\ast})$ starts with a trace of $\alpha$ during which $\neg \phi$ occurs at point $\nu_{k}(\zeta)$, namely at relative position $(k - (k_n + 1), \zeta)$. Let $\acute{\varsigma} \in \nu(\alpha)$ be this prefix of $\varsigma$. Consequently, $\acute{\varsigma} \models \langle \alpha \rangle \Diamond \neg \phi$ and the trace $\varrho \circ \acute{\varsigma}$ is a witness for $\sigma \models \langle \alpha^{\ast} \rangle \langle \alpha \rangle \Diamond \neg \phi$.
\end{itemize}
The proofs for $\langle ; \rangle \diamond$--$\langle {}^{\ast} \rangle \diamond$ are dual, since $\langle \alpha \rangle \Diamond \phi$ is equivalent to $\neg [\alpha]\Box \neg \phi$ by duality.
\end{proof}

\subsection{Proof of Relative Completeness}
\begin{proof}[Outline of Proof of Theorem~\ref{theo:completeness}]
The proof is a simple extension of the proof of Theorem 1 in~\cite{DBLP:conf/csl/Platzer10:TR}, which is the relative completeness theorem for \qdl, because the \qdtl calculus successively reduces temporal properties to non-temporal properties and, in particular, handles loops by inductive definition rules in terms of \qdl modalities. The temporal rules in Fig.~\ref{fig:rules} transform temporal formulas to simpler formulas, i.e., to where the temporal modalities occur after simpler programs ($[\cup]${\tiny $\Box$}, $[{}^{\ast}]${\tiny $\Box$}, $\langle \cup \rangle \diamond$, $\langle {}^{\ast} \rangle \diamond$) or disappear completely ($[?]${\tiny $\Box$}, $[:=]${\tiny $\Box$}, $[']${\tiny $\Box$} and $\langle ? \rangle \diamond$, $\langle := \rangle \diamond$, $\langle ' \rangle \diamond$). Hence, the inductive relative completeness proof for \qdl in~\cite{DBLP:conf/csl/Platzer10:TR} directly generalizes to \qdtl with the following addition: After applying $[{}^{\ast}]${\tiny $\Box$} or $\langle {}^{\ast} \rangle \diamond$, loops are ultimately handled by the standard \qdl rules {\em ind} and {\em con}. To show that sufficiently strong invariants and variants exist for the temporal postconditions $[\alpha]\Box\phi$ and $\langle \alpha \rangle \Diamond \phi$, we only have to show that such temporal formulas are expressible in the {\em first-order logic of quantified
differential equations}, FOQD~\cite{DBLP:conf/csl/Platzer10:TR}.
\end{proof}

\newpage
\section{Proofs for Liveness Verification by Quantifier Alternation}

In this section, we prove that the rules in Section 8 are sound.
\begin{proposition}[Local soundness]
The rules in Section 8 are locally sound for finitary liveness semantics.
\end{proposition}
\begin{proof}
Let $\sigma$ be any state.
\begin{itemize}
\item $[;] \diamond$ \ Assuming that the premiss is true, we need to consider two cases corresponding
to the two formulas of its succedent. If $\sigma \models [\alpha] \Diamond \phi$, then obviously $\sigma \models [\alpha;\beta] \Diamond \phi$, as every trace of $\alpha;\beta$ has a trace of $\alpha$ as prefix, during which $\phi$ holds at least once. If, however, $\sigma \models [\alpha][\beta]\Diamond \phi$,
then $\phi$ occurs at least once during all traces that start in a state reachable
from $\sigma$ by $\alpha$. Let $\varrho \circ \varsigma \in \mu(\alpha;\beta)$ with first\kern 0.25em$\varrho = \sigma$, $\varrho \in \mu(\alpha)$ and $\varsigma \in \mu(\beta)$. In finitary liveness semantics, $\varrho \circ \varsigma$ can be assumed to terminate (otherwise there is nothing to show). Then, last\kern 0.25em$\varrho$ is a state reachable from $\sigma$ by $\alpha$, hence $\varsigma \models \Diamond \phi$. Thus, $\varrho \circ \varsigma \models \Diamond \phi$.
\item $[\alpha]\diamond$ \ For the soundness of $[\alpha]\diamond$, first observe that the truth of $\sigma \llbracket \phi \rrbracket$ of $\phi$ depends
on the state $\sigma$, hence it can only be affected during state changes. Further,
the only actual changes of valuations happen during discrete jumps $\forall i : \kern -0.17em A \ f(\boldsymbol{s}) := \theta$ or continuous evolutions $\forall i : \kern -0.17em A \ f(\boldsymbol{s})' = \theta \with \chi$. All other system actions only cause control
flow effects but no elementary state changes. Assume the premise is
true in a state $\sigma$. If $\sigma \models \phi$, the conjecture is obvious. Hence, assume $\sigma \models \forall t : \kern -0.17em \mathbb{R}\ [\check{\alpha}]t = 1$. Suppose $\sigma \not \models [\alpha]\phi$; then there is a trace $\nu \in \mu(\alpha)$ with $\nu \not \models \Diamond \phi$. Then, this trace directly corresponds to a trace $\check{\nu}$ of $\check{\alpha}$ in which all $\phi \rightarrow t = 1$ conditions are trivially satisfied as $\phi$
never holds. As there are no changes of the fresh variable $t$ during $\check{\alpha}$, the
value of $t$ remains constant during $\check{\nu}$. But then we can conclude that there
is a trace, which is essentially the same as $\check{\nu}$ except for the constant valuation
of the fresh variable $t$ on which no conditions are imposed, hence $t = 0$ is possible. As these traces terminate in finitary liveness semantics, we can
conclude $\sigma \not \models \forall t : \kern -0.17em \mathbb{R} \ [\check{\alpha}] t = 1$, which is a contradiction. Conversely for
equivalence of premiss and conclusion, assume $\sigma \not \models \phi \vee \forall t : \kern -0.17em \mathbb{R} \ [\check{\alpha}] t = 1$. Then, the initial state $\sigma$ does not satisfy $\phi$ and it is possible for $\check{\alpha}$ to execute along a terminating trace $\nu$ that permits $t$ to be $\not = 1$. Suppose there was a position $(k, \zeta)$ of $\nu$ at which $\nu_{k}(\zeta) \models \phi$. Without loss of generality, we
can assume $(k, \zeta)$ to be the first such position. Then, the hybrid action which
regulates $\nu_k$ is accompanied by an immediate condition that $\varphi \rightarrow t = 1$,
hence $t = 1$ holds if $\nu$ terminates. Since the fresh variable $t$ is rigid (is
never changed during $\check{\alpha}$) and $\nu$ terminates in finitary liveness semantics, we conclude last\kern 0.25em$\nu \llbracket t \rrbracket = 1$, which is a contradiction.
\end{itemize}

\end{proof}

\end{document}